\documentclass[journal]{IEEEtran}
\usepackage{amsmath,amsfonts}
\usepackage{algorithmic}
\newtheorem{definition}{Definition}

\newtheorem{theorem}{Theorem}
\newtheorem{lemma}{Lemma}
\newtheorem{proof}{Proof}
\usepackage{booktabs}
\usepackage{array}
\usepackage[caption=false,font=normalsize,labelfont=sf,textfont=sf]{subfig}
\usepackage{textcomp}
\usepackage{stfloats}
\usepackage{url}
\usepackage{verbatim}
\usepackage{graphicx}
\usepackage{cite}
\usepackage{xcolor}
\usepackage[ruled,vlined]{algorithm2e}
\usepackage{amsmath, amssymb}

\SetKwInput{KwIn}{Input}
\SetKwInput{KwOut}{Output}
\SetKwComment{tcp}{\# }{}
\begin{document}

\title{Plasticity-Aware Mixture of Experts for Learning Under QoE Shifts in Adaptive Video Streaming}

\author{Zhiqiang He, \IEEEmembership{Member,~IEEE}, Zhi Liu, \IEEEmembership{Senior Member,~IEEE}

\thanks{This work of Zhiqiang He was supported by JST SPRING, Grant Number JPMJSP2131. The work of Zhi Liu was partially supported by ROIS NII Open Collaborative Research 2025-(252FA-23456), JSPS KAKENHI Grant Number 24K14927 and JST ASPIRE JPMJAP2432.}

\thanks{Zhiqiang He and Zhi Liu are with the Department of Computer and Network Engineering, the University of Electro-Communications, Japan. Email: hezhiqiang@ieee.org, liu@ieee.org.}

\thanks{The corresponding author is Zhi Liu.}
}

\markboth{Accepted by Journal: IEEE Transactions on Multimedia}%
{}


\maketitle

\begin{abstract}
Adaptive video streaming systems are designed to optimize Quality of Experience (QoE) and, in turn, enhance user satisfaction. However, differences in user profiles and video content lead to different weights for QoE factors, resulting in user-specific QoE functions and, thus, varying optimization objectives. This variability poses significant challenges for neural networks, as they often struggle to generalize under evolving targets—a phenomenon known as plasticity loss that prevents conventional models from adapting effectively to changing optimization objectives. To address this limitation, we propose the Plasticity-Aware Mixture of Experts (PA-MoE), a novel learning framework that dynamically modulates network plasticity by balancing memory retention with selective forgetting. In particular, PA-MoE leverages noise injection to promote the selective forgetting of outdated knowledge, thereby endowing neural networks with enhanced adaptive capabilities. In addition, we present a rigorous theoretical analysis of PA-MoE by deriving a regret bound that quantifies its learning performance. Experimental evaluations demonstrate that PA-MoE achieves a 45.5\% improvement in QoE over competitive baselines in dynamic streaming environments. Further analysis reveals that the model effectively mitigates plasticity loss by optimizing neuron utilization. Finally, a parameter sensitivity study is performed by injecting varying levels of noise, and the results align closely with our theoretical predictions. The code is available at https://github.com/tinyzqh/PA-MoE.
\end{abstract}


\section{Introduction}

\IEEEPARstart{A}daptive video streaming (AVS) has become the foundation of modern multimedia applications. To ensure user satisfaction, major platforms like Netflix and YouTube are continually improving playback quality. This improvement is guided by the concept of Quality of Experience (QoE), which incorporates key factors such as bitrate, rebuffering time, and latency \cite{10314007}. Traditionally, AVS systems optimize QoE by dynamically adjusting streaming parameters in response to changes in network conditions, ensuring a smooth viewing experience even with varying bandwidth. However, real-world deployments have shown that network bandwidth is not the only factor affecting QoE \cite{dobrian2011understanding, 10633863, 10185649}. Recent research has sought to improve QoE from the \textit{encoding}, \textit{management}, and \textit{assessment} perspectives, such as adaptive quantization for video encoding \cite{taha2023smart}, QoE-driven management systems \cite{taha2021qoe}, and automated QoE evaluation models \cite{taha2021automated}. However, these approaches primarily target system-level adaptation and do not fully capture the evolving nature of QoE. Recent evidence suggests that individual user preferences \cite{tang2023successor} and important regions within video content \cite{choi2024data} also play a role in QoE dynamics, making the optimization objective constantly changing. This presents a significant challenge: the optimization criteria are constantly shifting, requiring adaptive models that can track and respond to these changes in real time.  In summary, effectively addressing the AVS optimization problem in the face of dynamically changing QoE requires the development of robust, adaptive algorithms that can balance the demands of network variability and evolving user-centric factors \cite{10314007}.

Classical methods like Model Predictive Control\cite{bayesmpc} struggle with changing objectives. Existing research on QoE dynamics can be broadly divided into two stages. The first stage focuses on the user level, where factors such as user preferences \cite{tang2023successor} and regions of interest within the video \cite{choi2024data,8761156} are taken into account. In this stage, data including watch time, likes, and timestamps are utilized to predict user interests, which then inform the adjustment of QoE parameters \cite{lee2022qrator, kimura2021context}. The second stage builds on these predictions to optimize QoE, enabling the streaming system to adapt to dynamically changing user demands \cite{zhang2021sensei, zhang2022enabling}. However, two-stage optimization methods encounter several challenges that hinder effective QoE optimization. Inaccuracies in user-level predictions can cascade into subsequent stages, resulting in suboptimal outcomes. Furthermore, their sequential nature, dependence on pre-trained models, and reliance on historical data delay real-time adjustments and updates, while also limiting the ability to generalize when faced with QoE variations beyond the training distribution. Existing approaches for handling dynamic QoE weights, such as multi-objective Q-networks \cite{abels2019dynamic} and meta-reinforcement learning \cite{kan2022improving}, typically rely on a key assumption: the dynamics of the weight changes are known in advance and can be used either as input to the model \cite{abels2019dynamic} or for pretraining \cite{kan2022improving}. In contrast, our method removes this critical assumption by directly addressing the problem from the perspective of plasticity, making our approach more aligned with realistic industrial scenarios. Some researchers have even posited that supporting a large number of QoE objectives is inherently infeasible \cite{9796953}. In this work, we challenge that notion by leveraging a neural network capable of effectively managing a broad spectrum of QoE objectives.

The impact of QoE variations on the optimization objective is evident when considering that different content types have distinct QoE preferences. For instance, while news broadcasts may prioritize stable playback even at lower bitrates, sports events demand higher frame rates and minimal stuttering for an optimal viewing experience. To address these challenges, we propose a unified, single-policy learning approach that overcomes the limitations of two-stage methods by directly adapting to varying QoE objectives within one framework. Nevertheless, employing a single-stage neural network—where a single set of parameters must rapidly learn and adjust to multiple dynamic objectives—presents a formidable challenge for deep neural network–based AVS optimization. Similar to human neurons \cite{puderbaugh2023neuroplasticity}, neural networks are susceptible to plasticity loss; frequent switching between changing objectives disrupts established representations and hinders the retention of previously acquired knowledge \cite{dohare2024loss}. This degradation in plasticity undermines the model’s ability to generalize and efficiently adapt to new QoE shifts, complicating the overall optimization process \cite{LyleZNPPD23, gulcehre2022an, Hare_Tortoise}. Fortunately, existing research demonstrates that Mixture of Experts (MoE) architectures can effectively mitigate this issue by preserving network plasticity through the dynamic allocation of specialized experts for different tasks \cite{moe_rl}. Moreover, the injection of noise into neural networks has been shown to counteract plasticity loss and enhance overall performance by improving the model's robustness to shifting objectives \cite{ash2020warm, d2022sample, nikishin2022primacy, nauman2025bigger}.

Despite significant advancements, current approaches still have notable limitations in dynamically adapting QoE optimization in AVS. While standard mixture-of-experts (MoE) architectures help mitigate plasticity loss, they lack explicit mechanisms to balance forgetting with retention. This can lead to scenarios where outdated experts retain irrelevant knowledge while over-specialized experts struggle to generalize across new QoE objectives \cite{sun2022paco,hendawy2024multitask}. To address these challenges, we propose Plasticity-Aware Mixture of Experts (PA-MoE), a novel learning framework designed to regulate neural plasticity in response to dynamic QoE shifts. Our approach draws inspiration from human adaptive behavior, where selective suppression of irrelevant memories \cite{anderson2001suppressing} facilitates adjustment to changing environments. PA-MoE integrates a controlled forgetting mechanism via noise injection \cite{nikishin2022primacy}, systematically discarding obsolete knowledge while preserving essential shared information through adaptive weight adjustments. By striking an optimal balance between adaptability and stability, PA-MoE ensures efficient learning and robust performance across a wide range of QoE optimization objectives in AVS. Furthermore, our approach is supported by a rigorous theoretical analysis that substantiates its efficacy.

In summary, this paper makes the following contributions:

\begin{list}{$\bullet$}{\setlength{\leftmargin}{10.0pt} \setlength{\itemindent}{0.pt} \setlength{\labelsep}{0.5em}}
  \item \textbf{Novel Problem Analysis:} We are the first to identify the loss of neural network plasticity induced by fluctuations in QoE, providing a new perspective on optimizing systems for shifts in QoE due to changes in user preferences and video content.
  \item \textbf{Algorithm Design:} We propose PA-MoE, a novel algorithm that incorporates controlled noise injection into experts to actively eliminate outdated knowledge, thus maintaining the adaptive capability of the MoE framework.
  \item \textbf{Theoretical Foundation:} We provide a rigorous theoretical analysis of PA-MoE by deriving its performance regret bound and demonstrating how noise injection dynamically balances the trade-off between forgetting and memory retention.
  \item \textbf{Implementation and Evaluation:} Extensive experiments show that PA-MoE rapidly adapts to QoE shifts, achieving a 45.5\% performance improvement over conventional MoE models. Comprehensive analysis across multiple metrics elucidates the mechanisms behind PA-MoE's effectiveness in adapting to dynamic system changes, while a sensitivity study on noise injection further validates our theoretical framework.
\end{list}

\section{Background And Motivation}

\subsection{Problem Formulation}

In adaptive video streaming \cite{kan2022improving}, a video is typically divided into $I$ consecutive chuncks $\{V_1, V_2, \dots, V_I\}$, with each chunk lasting $t_{i}$ seconds. These chunks are encoded at multiple bitrates $b \in \mathbb{B} = \{b_1, b_2, ..., b_J\}$ in ascending order, such that the size of chunk $i$, denoted as $d_{i}^{b}$, increases with the bitrate $b$. The goal is for a reinforcement learning agent to select the optimal bitrate for each chunk in order to maximize QoE. To prevent interruptions during transitions between segments $V_{i}$ and $V_{i+1}$, a playback buffer $B_{f}$ is used to prefetch upcoming segments, allowing for simultaneous playback and downloading. The buffer occupancy, $B_{f}(t)$, is limited to the interval $[0, B_f^{max}]$ to account for server or client limitations.

The download time for segment $V_{i}^{b}$ is composed of two main components: the network round-trip time $t_{RTT}$ and data transfer time $d_{i}^{b}/S_{i}$ represents the download speed. In order to account for network variability, a noise factor $\eta(i)$ is introduced. As a result, the download start time for segment $V_{i+1}$ can be calculated as $t_{i+1}^{d} = t_{i}^{d} + \left( \frac{d_{i}^{b}}{S_{i}} + t_{RTT} \right) \times \eta(i) + \Delta t_{i}^{w}$, where $\Delta t_i^w$ is a deliberate waiting period implemented to manage buffer constraints. The evolution of buffer occupancy is described by:

\begin{equation}
    B_{f}(t_{i+1}) = \left( \left(B_{f}\left(t_{i} \right) - t_{\text{delay}} \right)_{+} + t_{i} - \Delta t_{i}^{w}  \right)_{+}, 
\end{equation}
with the network delay defined as $t_{\text{delay}} = \left( \frac{d_{i}^{b}}{S_{i}} + t_{RTT} \right) \times \eta(i)$. To prevent buffer overflow, the waiting time is determined by
\begin{equation}
    \Delta t_{i}^{w} = \left( \left(B_{f}\left(t_{i} \right) - t_{\text{delay}} \right)_{+} + t_{i} - B_{f}^{max}  \right)_{+}.
\end{equation}

In adaptive video streaming, the central objective is to optimize the QoE. A typical QoE model incorporates three critical components. The first component is the \textbf{BitRate} factor, which measures the average video quality, given by $\frac{1}{I} \sum_{i=1}^{I}q(V_{i}^{b})$, where $q(\cdot)$ is a monotonically increasing quality function; The second component is the \textbf{Smooth} factor,  which penalizes quality fluctuations. This is represented by the term $\frac{1}{I-1} \sum_{i=1}^{I-1}|q(V_{i+1}^{b}) - q(V_{i}^{b})|$. This captures the perceptual impact of abrupt changes in quality. The third component is the \textbf{Rebuffer} penalty, which accounts for playback interruptions due to buffer underruns caused by network delays and transmission time. These factors are combined into a weighted sum to balance video quality, smoothness, and playback stability:

\begin{align}
    \begin{split}
        QoE^{I} = &\mu_{1} \sum_{i=1}^{I}  q(V_{i}^{b}) - \mu_{2} \sum_{i=1}^{I-1}|q(V_{i + 1}^{b}) - q(V_{i}^{b})| - \\
        &  \mu_{3} \sum_{i=1}^{I} \left( \left( \frac{d_{i}^{b}}{S_{i}} + t_{RTT} \right) \times \eta(i) - B_{f}(t_{i}) \right)_{+}.
    \end{split}
\end{align}

Accordingly, the overall problem is naturally formulated as a sequential decision-making process, which can be modeled as a Markov Decision Process (MDP). This formulation enables the application of state-of-the-art reinforcement learning algorithms to effectively address the QoE optimization challenge in adaptive video streaming.

The state and action spaces are aligned with mainstream approaches such as Pensieve \cite{mao2017neural}. The environment follows a standard reinforcement learning formulation: the agent interacts with the streaming system by observing network and playback states and selecting bitrate actions accordingly. The input is the state $s_t \in \mathbb{R}^{6 \times 8}$, comprising six categories of information, each tracked over the most recent eight time steps: (i) normalized last selected bitrate; (ii) normalized buffer occupancy; (iii) measured throughput (downloaded size per unit time); (iv) normalized delay; (v) sizes of the next video chunk for all bitrate levels; and (vi) remaining number of chunks until the end of the video. PA-MoE outputs a discrete bitrate action $a_t \in \{0,1,2,3,4,5\}$, where each value corresponds to one of the six available bitrate levels in the streaming system.

\subsection{Markov Decision Process With Shift QoE}
An MDP is formally defined as a tuple $(\mathcal{S}, \mathcal{A}, \mathcal{P}, \mathcal{R}, \gamma)$, where $\mathcal{S}$ represents the state space containing all the pertinent information for decision-making, and $\mathcal{A}$ denotes the action space comprising all feasible decisions. The transition function $\mathcal{P}: \mathcal{S} \times \mathcal{A} \rightarrow \Delta(\mathcal{S})$ captures the stochastic evolution of states, while the reward function $\mathcal{R}: \mathcal{S} \times \mathcal{A} \rightarrow \mathbb{R}$ delivers QoE-aligned feedback to guide the learning process. The discount factor $\gamma \in [0,1)$ is employed to balance short-term and long-term rewards.

In our formulation, the state space $\mathcal{S}$ includes important streaming indicators such as the quality of the last downloaded chunk, buffer occupancy, network throughput and delay, sizes of upcoming chunks, and the number of remaining chunks. This comprehensive representation allows the agent to make informed decisions based on the context. The action space $\mathcal{A}$ consists of multiple discrete bitrate levels, each corresponding to a different quality version. Lastly, the reward $\mathcal{R}$ is directly determined by the QoE metric, effectively balancing the trade-offs between video quality, smoothness, and playback stability.

By utilizing reinforcement learning (RL), the system is able to adaptively learn optimal control strategies, even in uncertain environments. This effectively addresses the challenges posed by dynamic and unpredictable system behavior. However, when there are shifts in QoE due to variations in video content, the reward function experiences significant changes across the state space. This variability transforms the problem into a non-stationary MDP, which can be seen as a sequence of distinct MDPs $\{M_1, M_2, ..., M_N\}$. Each MDP $M_n = (\mathcal{S}, \mathcal{A}, \mathcal{P}, \mathcal{R}_n, \gamma)$ represents a unique optimization problem defined by its own reward structure. This non-stationarity poses a fundamental challenge for conventional learning algorithms, as policies optimized for one network environment may fail catastrophically when deployed in another. Therefore, there is a critical need for adaptive strategies that can generalize across these dynamic and evolving conditions.

\subsection{Impact of Shift QoE on RL}
To demonstrate the impact of QoE shifts on reinforcement learning–based optimization and the resulting loss of neural plasticity, we conducted an experimental case study. Specifically, we designed three distinct video types—documentaries (D), live streams (L), and news (N)—each with different coefficients for the QoE components in our system model. These coefficients are represented as ${\mu_{1}, \mu_{2}, \mu_{3}}$, where the value $\beta = 6$ indicates a strong emphasis on a particular QoE component. The coefficient configurations for each video type are as follows: Documentaries (D): ${\mu_{1}=1, \mu_{2}=\beta, \mu_{3}=1}$; Live streams (L): ${\mu_{1}=1, \mu_{2}=1, \mu_{3}=\beta}$; News (N): ${\mu_{1}=\beta, \mu_{2}=1, \mu_{3}=1}$. Figure~\ref{changing_qoe_action} illustrates the action outputs of reinforcement learning agents using the Proximal Policy Optimization (PPO) \cite{schulman2017proximal} algorithm across various network architectures as the system conditions evolve. The variation of QoE and the rationale for selecting Action as the evaluation metric are provided in Appendix \ref{shift_qoe_qoe_problem}.

\begin{figure}[ht]
  \centering
  \includegraphics[width=2.3in]{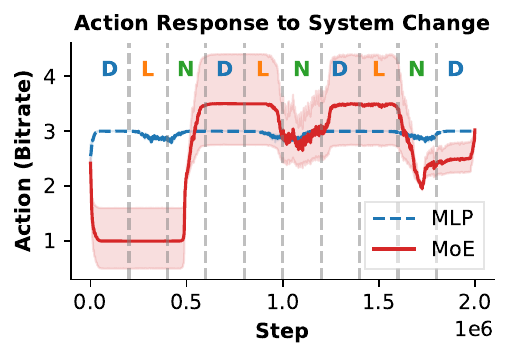}
  \caption{Action output variation across different network architectures under shifted QoE reward conditions. D, L, and N represent distinct QoE metrics. An ideal agent adapts its actions as QoE changes, but the adaptability varies with network design.}
  \label{changing_qoe_action}
\end{figure}

This example highlights the impact of changes in QoE parameters on the adaptability of learning agents. Specifically, agents using a multilayer perceptron (MLP) architecture struggle to quickly adjust their actions in response to QoE changes, while those utilizing a mixture-of-experts (MoE) architecture show a much stronger ability to adapt to evolving QoE. However, the changes in actions do not follow a clear pattern that directly corresponds to the periodic shifts in QoE.

\section{Related Work}

We evaluate PA-MoE in comparison with related approaches along three key dimensions: adaptive video streaming under varying QoE metrics, MDP adaptation and optimization in video streaming, and the application of Mixture of Experts methods in policy learning.

\subsection{Content Aware Video Streaming}

QoE metrics are broadly classified into pixel-based quality and streaming quality \cite{dobrian2011understanding}. Pixel-based quality evaluates intrinsic visual fidelity—such as sharpness, color accuracy, and contrast—independently of network-induced distortions, whereas streaming quality encompasses dynamic factors including playback smoothness, quality degradation, buffering frequency, and user-specific sensitivity to quality variations \cite{zhang2021sensei}. In this work, we concentrate on QoE variations driven by streaming quality.

Conventional approaches typically estimate QoE by inferring users' regions of interest within a video \cite{8761156}. Some methods utilize deep learning models, such as CNNs and Transformers \cite{choi2024data}, while others rely on engagement signals, such as timestamps, likes, and watch duration, to approximate user preferences \cite{kimura2021context}. However, prior studies have primarily focused on modeling QoE variations within a single video type. In contrast, our study investigates QoE shifts across different video types \cite{9699071}, providing a broader perspective on adaptive streaming optimization.

Once user-preference-aware QoE is estimated, traditional techniques optimize adaptive streaming through buffer-based strategies \cite{lee2022qrator, kimura2021context}, reinforcement learning \cite{tang2023successor, 9796953}, apprenticeship learning \cite{9699071}, or user-specific models \cite{zhang2022enabling}. These two-stage optimization methods introduce additional computational overhead, and the performance of the first stage does not have a clear correlation with the overall performance \cite{he2025understanding}. In this paper, we propose a unified policy that dynamically modulates network memory and plasticity through noise injection, thereby enabling rapid adaptation to evolving QoE patterns.

\subsection{Video Streaming with QoE Shifts}

If a single policy is used to handle multiple MDPs simultaneously, it falls under the category of Multi-Task Learning (MTL) \cite{WilliOFDC24}. On the other hand, training a policy in a streaming manner to adapt to continuously changing MDPs is classified as Continual Learning \cite{AbbasZM0M23}. For example, concurrently training a policy to handle $M_{1}$ and $M_{2}$ would be considered to MTL, while training sequentially on $M_{1}$ for a set number of steps, then on $M_{2}$ and cycling through (e.g., $M_1 \to M_2 \to M_1 \to M_2$) would be classified as Continual Learning. In our work, we focus on streaming tasks that involve dynamic QoE shifts in video streaming, which can be modeled as a sequence of MDPs (e.g., $M_1 \to M_2 \to \cdots, \to M_1 \to \cdots$) \cite{elsayed2024streaming}.

Modeling system changes as a Continual Learning problem is a common approach \cite{zhang2022towards}. However, existing methods typically decompose the challenge. For example, some methods first employ a monitor to detect system changes before applying adjustments such as retraining \cite{zhang2022towards}; others train a QoE identifier \cite{wu2024mansy}  and a network-bandwidth classifier \cite{zhang2025novel}; and in meta-reinforcement learning, an autoencoder is used to compress prior knowledge about network bandwidth \cite{kan2025merina+}. Meta reinforcement learning enables systems to "learn how to learn" across multiple tasks \cite{wang2024mmvs, li2023metaabr}, facilitating rapid adaptation. However, this approach often assumes that tasks share certain stationary properties in their distributions, which is often unrealistic in real-world adaptive video streaming applications.

In contrast, our approach does not rely on a two-stage optimization or require prior knowledge of task relationships or distributions. Instead, we directly adjust the plasticity of neural networks through noise injection, enabling them to dynamically adapt to varying optimization objectives based on the data.

\subsection{Mixture of Experts}

The MoE architectures utilizes a router mechanism to dynamically assign input states to specialized expert networks for decision-making \cite{cai2024survey}. This approach has been widely adopted in large-scale models \cite{locmoe_ijcai24,MoELLaVA,openmoe} and has shown exceptional performance in multi-task learning and Continual Learning (CL) scenarios \cite{li2025theory}. Therefore, it is natural and intuitive to consider applying it to handle the objective shifts in AVS. However, a key challenge in MoE is effectively allocating tasks to experts. To address this, various strategies have been explored, including linear programming \cite{lewis2021base}, reinforcement learning \cite{roller2021hash}, optimal transport \cite{liu2023sparsity}, and greedy top-k token selection \cite{zhou2022mixture}. Recent advancements, such as Soft MoE, have introduced slot-based connectivity to allow for more flexible integration of experts \cite{puigcerver2024from}. Other approaches have also explored parameter composition to optimize task-specific expert selection \cite{sun2022paco}. In our video streaming task, even the basic Sparse MoE (SMoE) can efficiently allocate tasks to achieve load balancing. To better investigate the problem and eliminate the influence of other factors, we opted for the simplest form of SMoE \cite{shazeer2017}.

Our approach differs from existing strategies that primarily aim to reduce computational overhead during inference. Instead, we focus on enhancing MoE adaptability, which is crucial for real-world applications where the policy network needs to quickly adjust to new tasks. To achieve this, we introduce a knowledge forgetting module that selectively discards outdated information to preserve plasticity. Additionally, our memory mechanism allows the neural network to autonomously adjust its weights, effectively counterbalancing the effects of active forgetting. Previous research has mainly focused on knowledge sharing, but methods such as direct parameter copying or weight sharing can result in negative transfer, as not all information is beneficial to share \cite{hendawy2024multitask}. In contrast, our approach exclusively injects forgetting, giving the network the ability to independently retain valuable memories.

Moreover, recent studies have shown that RL policy networks often have a significant number of underutilized parameters, with around 90\% of neuron outputs remaining inactive \cite{sokar2023dormant}. Injecting noise into network parameters has found to be effective in maintaining plasticity and improving overall performance \cite{sokar2023dormant, schwarzer2023bigger, doro2023sampleefficient}. While noise perturbations are commonly used in traditional dense architectures \cite{rl_pruned_network, nikishin2022primacy}, our work maintains the plasticity of MoE models by injecting noise into the gradients. While noise is commonly added in reinforcement learning to encourage exploration \cite{plappert2018parameter}, our method differs in two key aspects: (1) noise is injected continuously rather than only at the beginning of each episode; and (2) the objective is to enhance plasticity rather than exploration, offering a new perspective on Mixture-of-Experts (MoE) models and plasticity loss.  Additionally, we provide a theoretical analysis of the regret bound associated with our proposed PA-MoE method to support our approach.

\section{Proposed Method}

\subsection{Architectural Paradigm of the MoE Framework}

The challenge posed by QoE shifts is that they render the underlying problem non-stationary, leading to a loss of agent plasticity and hindering rapid adaptation for dynamic knowledge transfer \cite{sun2022paco, hendawy2024multitask}. In order to address these issues, we propose the PA-MoE method, which directly injects plasticity into the experts of the MoE architecture. In a typical MoE model, there are $N$ experts, a router, and a gating network—usually implemented as a series of linear layers \cite{fedus2022switch}.

\begin{figure}[!htbp]
\centering
\includegraphics[width=3.5in]{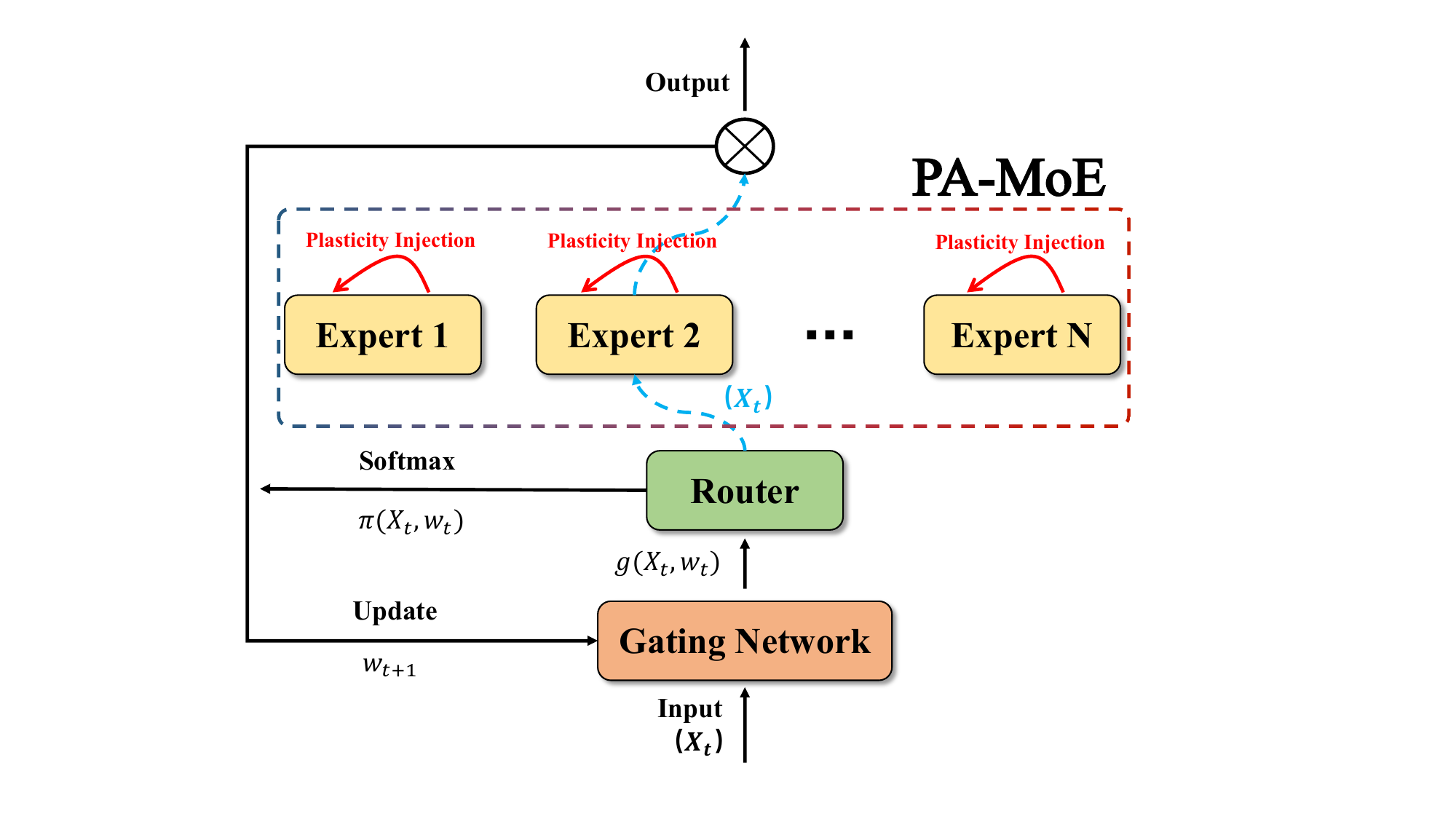}%
\caption{An illustration of the PA-MoE model.}
\label{pa_moe}
\end{figure}

The framework of PA-MoE is illustrated in Figure~\ref{pa_moe}. Once the MDP is formulated, the next step is to identify a mapping from input states to actions, known as the policy. At time step $t$, the input state vector is denoted by $\boldsymbol{X}_{t} \in \mathbb{R}^{d}$. The gating network applies a linear transformation to $\boldsymbol{X}_{t}$, producing an output $g_{n}(\boldsymbol{X}_{t}, \boldsymbol{\boldsymbol{w}}_{t}^{(n)})$ for each expert $n \in \left[ N \right]$, where $\boldsymbol{\boldsymbol{w}_{t}^{(n)}} \in \mathbb{R}^{n}$ represents the parameters corresponding to expert $n$. The complete set of gating network parameters is represented as $\boldsymbol{\boldsymbol{w}_{t}}:=\left[ \boldsymbol{\boldsymbol{w}}_{t}^{1}, \cdots, \boldsymbol{\boldsymbol{w}}_{t}^{(N)} \right]$, and the overall gating output is defined as

\begin{align}
    \boldsymbol{g}(\boldsymbol{X}_{t}, \boldsymbol{\boldsymbol{w}}_{t}):= \left[ g_{1}(\boldsymbol{X}_{t}, \boldsymbol{\boldsymbol{w}}_{t}^{(1)}), \cdots, g_{N}(\boldsymbol{X}_{t}, \boldsymbol{\boldsymbol{w}_{t}^{(N)}}) \right],
\end{align}
which can equivalently be expressed as $\boldsymbol{g}(\boldsymbol{X}_{t}, \boldsymbol{\boldsymbol{w}}_{t})=\boldsymbol{\boldsymbol{w}}_{t}^{T} \boldsymbol{X}_{t}$. To enhance network sparsity and reduce computational costs, our router selects only the Top-1 expert—that is, the expert corresponding to the maximum gating output—as demonstrated in \cite{theory_mixtur, fedus2022switch}. Although the architecture includes $N$ experts, only one expert is selected for each input, keeping the per-sample complexity at $O(1)$. Furthermore, to promote exploration, noise is injected into the router following the methods outlined in \cite{fedus2022switch, li2025theory}. Specifically, the chosen expert at time step $t$ is determined by

\begin{align}
    n_{t} = \arg \max_{n} \left\{g_{n}(\boldsymbol{X}_{t}, \boldsymbol{\boldsymbol{w}}_{t}^{(n)}) + z_{t}^{(n)} \right\},
\end{align}
where $z_{t}^{(n)}$ is a noise term drawn from a Gaussian distribution for each expert $n$. This stochastic component promotes exploration by perturbing the gating output. Furthermore, the final output of the gating network is computed using a softmax function

\begin{align}
    \boldsymbol{\pi}_n \left(\boldsymbol{X}_t, \boldsymbol{\boldsymbol{w}}_t\right)=\frac{\exp \left(g_n\left(\mathbf{X}_t, \boldsymbol{\boldsymbol{w}}_t^{(n)}\right)\right)}{\sum_{n^{\prime}=1}^N \exp \left(g_{n^{\prime}}\left(\mathbf{X}_t, \boldsymbol{\boldsymbol{w}}_t^{(n)}\right)\right)}, \quad \forall n \in[N].
\end{align}
This softmax output represents a probability distribution over the experts and is subsequently used to update the gating network parameters $\boldsymbol{\boldsymbol{w}}_{t+1}$ for all experts in the MoE model.

\subsection{Training Expert With Plasticity Injection}

In this section, we will introduce the PA-MoE method for plasticity injection. Random perturbations and resets of neural network parameters have been shown to be effective strategies for mitigating plasticity loss in reinforcement learning \cite{galashov2024nonstationary, shin2024dash}. In our PA-MoE framework, we enhance plasticity by injecting controlled noise directly into the parameter updates. To validate the effectiveness of this approach, we integrate it into a state-of-the-art reinforcement learning algorithm, PPO. 

Let $\boldsymbol{w}_t^{(n)}$ denote the parameters of expert $n$ in the actor network at the $t$-th training step. The policy distribution at time step $t$ is then given by:

\begin{align}
     \pi_{\boldsymbol{w}}(a_{t} | X_{t}) = \sum_{n=1}^{N} \boldsymbol{\pi}_n \left(\boldsymbol{X_{t}, \boldsymbol{\boldsymbol{w}_{t}}}\right) w_{t}^{(n)} (\boldsymbol{X_{t}}).
\end{align}
For the critic, the update rule is defined as

\begin{align}
    L^{critic}\left( \boldsymbol{w} \right) = \mathbb{E}_{t} \left[ \left( V_{\boldsymbol{w}}(X_{t}) - \hat{V}_{t} \right)^{2} \right],
\end{align}
where, $\hat{V}_{t}$ represents the target value, typically computed using Generalized Advantage Estimation (GAE, $\hat{A}_{t}$) or a Temporal Difference (TD) target. PPO further employs a clipped surrogate objective defined as

\begin{align}
    L_t(\boldsymbol{w})\!=\!\mathbb{E}_t \!\left[\min \! \left(r_t(\boldsymbol{w}) \hat{A}_t, \!\operatorname{clip}\left(r_t(\boldsymbol{w}), 1\!-\!\epsilon\!, 1\!+\!\epsilon\right)\! \hat{A}_t \! \right)\right], 
\end{align}
with the probability ratio
\begin{align}
    r_t(\boldsymbol{w})= \pi_{\boldsymbol{w}}\left(a_t \mid s_t\right) / \pi_{\boldsymbol{w}_{\text {old }}}\left(a_t \mid s_t\right).
\end{align}

For each normal expert $i$ (selected via the Top-K mechanism), we perform a noisy gradient descent update. Specifically, the update at time $t$ is given by

\begin{equation}
    \label{ppoloss}
    \boldsymbol{w}_i^{t+1}=\boldsymbol{w}_i^t-\eta \nabla L_t\left(\boldsymbol{w}_i^t\right)+\eta \gamma \epsilon_i^t, \quad t=1,2, \ldots, T,
\end{equation}
where $\eta > 0$ is the learning rate, assumed to satisfy $\eta \leq 1/L$, $\gamma > 0$ controls the noise injection strength, and $\epsilon_{i}^{t} \sim \mathcal{N}(0, I_{d})$ is a standard Gaussian noise vector in $\mathbb{R}^{d}$, independent across time steps and experts. This noise injection facilitates a form of controlled forgetting by discarding outdated information, thereby enhancing the model’s ability to adapt to nonstationary environments. We provide a theoretical analysis of this mechanism below.

In PPO, the policy update is constrained within a trust region, which ensures that each update only makes small changes to the current parameters. Although the global objective may be highly nonconvex, within the trust region, the loss function can be well approximated locally by a quadratic function. This approximation allows us to make certain assumptions about the behavior of the loss function, such as being L-smooth and $\mu$-strongly convex, which is also considered in \cite{agarwal2021theory}. These assumptions are useful for theoretical analysis and for understanding how the parameters adapt to a changing environment.

Let $L_{t}(\boldsymbol{w})$, with $\boldsymbol{w} \in \mathbb{R}^{d}$, denote each expert's local objective at time $t$. We consider the following standard assumptions:

\textbf{(A1) L-Smoothness:} For all $\boldsymbol{w}, \boldsymbol{w}^{\prime} \in \mathbb{R}^{d}$, the gradient of $L_{t}$ is Lipschitz continuous, $\left\|\nabla L_t(\boldsymbol{w})-\nabla L_t\left(\boldsymbol{w}^{\prime}\right)\right\| \leq L\left\|\boldsymbol{w}-\boldsymbol{w}^{\prime}\right\|$, where $L > 0$ is the smoothness constant.

\textbf{(A2) $\mu$-Strong Convexity:} For all $\boldsymbol{w}, \boldsymbol{w}^{\prime} \in \mathbb{R}^{d}$, the function $L_{t}$ satisfies, $L_t\left(\boldsymbol{w}^{\prime}\right) \geq L_t(\boldsymbol{w})+\nabla L_t(\boldsymbol{w})^{\top}\left(\boldsymbol{w}^{\prime}-\boldsymbol{w}\right)+\frac{\mu}{2}\left\|\boldsymbol{w}^{\prime}-\boldsymbol{w}\right\|^2$, with $\mu > 0$ being the strong convexity constant.

\textbf{(A3) Nonstationarity:} Let the time-varying optimal parameter be defined as $\boldsymbol{w}_{t}^{*}=\arg \min_{\boldsymbol{w}} L_t(\boldsymbol{w})$.
Define the path length of the optimal parameters over $T$ rounds as $P_T=\sum_{t=1}^{T-1}\left\|\boldsymbol{w}_{t+1}^*-\boldsymbol{w}_t^*\right\|$. We assume that $P_{T} < \infty$;  this quantity measures the rate at which the environment (or the optimum) changes over time.

We now state two key lemmas that are instrumental for our analysis.

\begin{lemma}[\textbf{Gradient Co-coercivity Lemma}]\label{GradientCocoercivityLemma} 
Let $L_t:\mathbb{R}^d\to\mathbb{R}$ be an $L$–smooth function; that is, for all $\boldsymbol{w},\boldsymbol{w}'\in\mathbb{R}^d$,
\[
\|\nabla L_t(\boldsymbol{w})-\nabla L_t(\boldsymbol{w}')\| \le L\|\boldsymbol{w}-\boldsymbol{w}'\|.
\]
Then, for any $\boldsymbol{w}_i^t,\boldsymbol{w}_t^*\in\mathbb{R}^d$, we have
\begin{align}\label{eq:main-ineq}
\|\nabla L_t(\boldsymbol{w}_i^t) & -\nabla L_t(\boldsymbol{w}_t^*)\|^2  \\ 
& \le L\, (\nabla L_t(\boldsymbol{w}_i^t)-\nabla L_t(\boldsymbol{w}_t^*))^{T} (\boldsymbol{w}_i^t-\boldsymbol{w}_t^*). \nonumber
\end{align}
\end{lemma}
\begin{proof}
    See Appendix \ref{appendix_proof_gradient_co_coercivity_lemma}.
\end{proof}

\begin{lemma}[\textbf{Gradient Strong Convexity Lemma}]\label{GradientStrongConvexity}
Let $L_t:\mathbb{R}^d\to\mathbb{R}$ be a $\mu$–strongly convex function; that is, for all $\boldsymbol{w},\boldsymbol{w}'\in\mathbb{R}^d$,
\[
L_t(\boldsymbol{w}') \ge L_t(\boldsymbol{w}) + \nabla L_t(\boldsymbol{w})^{T} (\boldsymbol{w}'-\boldsymbol{w}) + \frac{\mu}{2}\|\boldsymbol{w}'-\boldsymbol{w}\|^2.
\]
Then, for any $\boldsymbol{w}_i^t,\boldsymbol{w}_t^*\in\mathbb{R}^d$, it holds that
\begin{equation}\label{eq:strong-convexity}
\left(\nabla L_t(\boldsymbol{w}_i^t) - \nabla L_t(\boldsymbol{w}_t^*)\right)^{T} (\boldsymbol{w}_i^t-\boldsymbol{w}_t^*) \ge \mu \|\boldsymbol{w}_i^t-\boldsymbol{w}_t^*\|^2.
\end{equation}
\end{lemma}
\begin{proof}
    See Appendix \ref{appendix_proof_strongconvexity_lemma}.
\end{proof}

\begin{theorem}[\textbf{Tracking Error Bound under Nonstationarity}]\label{thm:error-bound}
Under Assumptions (A1)--(A3) and with the update rule
\begin{equation}\label{eq:update-rule}
\boldsymbol{w}_i^{t+1} = \boldsymbol{w}_i^t - \eta\, \nabla L_t\bigl(\boldsymbol{w}_i^t\bigr) + \eta\, \gamma\, \epsilon_i^t, \quad \epsilon_i^t \sim \mathcal{N}\left(0, I_d\right),
\end{equation}
there exists a constant $C>0$ (depending only on $L$ and $\mu$) such that the average squared error satisfies
\begin{equation}\label{eq:main-bound}
\frac{1}{T} \sum_{t=1}^{T} \mathbb{E}\|e_t\|^2 \le C \left(\eta \gamma^2 d + \frac{P_T^2}{T \eta}\right),
\end{equation}
where the error is defined as $e_t = \boldsymbol{w}_i^t - \boldsymbol{w}_t^*$ and the path length of the optimal parameters is $P_T = \sum_{t=1}^{T-1}\|\boldsymbol{w}_{t+1}^* - \boldsymbol{w}_t^*\|$.

\end{theorem}
\begin{proof}
    See Appendix \ref{appendix_proof_distill_enhance}.
\end{proof}

We will now analyze Theorem \ref{thm:error-bound} from the perspectives of active forgetting and memorization. \textbf{Active Forgetting:} The term $\eta \gamma^{2} d$ in the error bound increases linearly with the learning rate $\eta$. A higher $\eta$ implies that each expert updates its parameters more aggressively in response to new information, effectively promoting active forgetting - rapidly discarding outdated information in favor of recent data. However, if $\eta$ is set too high, the amplified noise (scaled by $\gamma$) can render the updates overly volatile, potentially destabilizing the model. Therefore, while a larger $\eta$ facilitates swift adaptation, it must be chosen carefully to avoid excessive noise-induced fluctuations. \textbf{Memorization:} On the other hand, the term $\frac{P_{T}^{2}}{T \eta}$ is inversely proportional to $\eta$. This means that a smaller $\eta$ will result in slower updates of the parameter, causing the model to hold onto outdated information. This can be thought of as the model "memorizing" past data. However, this excessive retention hampers the expert's ability to track the current optimal parameters, particularly when the environment is nonstationary and the optimum drifts over time (as captured by the path length $P_{T}$). Therefore, an overly small $\eta$ results in insufficient active forgetting, leading to larger tracking errors.

In summary, excessive noise can cause the neural network's updates to be overly influenced by the noise, resulting in a biased model that is prone to forgetting. Conversely, insufficient noise can lead to the updates be dominated by $\frac{P_{T}^{2}}{T \eta}$, causing outdated information to become fixed in the model. Therefore, it is crucial to properly control both the noise level and the learning rate. This requires striking a balance between having a large enough noise level to actively forget outdated information and adapt to new conditions, while also keeping the noise-induced error under control.

While the convergence proof assumes $L$-smoothness and $\mu$-strong convexity for analytical convenience, these are intended as local regularity conditions within each PPO step rather than claims of global convexity; when steps are large or gradients are indistinguishable from noise, the local quadratic approximation can break down. Nevertheless, PA-MoE exhibits stable convergence under nonconvex neural losses in practice, suggesting that its plasticity-aware updates remain effective beyond these idealized assumptions.

\section{Experiments}
To the best of our knowledge, our proposed method is the first that does not require any prior knowledge about the non-stationarity of the environment. While our main focus is on analyzing the problem itself and exploring solutions that do not rely on any prior assumptions, we additionally include comparisons with approaches that make use of prior knowledge (e.g., meta-reinforcement learning) to provide a more comprehensive evaluation. In our experiments, we aim to address the following key questions from three perspectives: \textbf{Cause}, \textbf{Effect}, and \textbf{Analysis}: (1.1) Can PA-MoE effectively mitigate the problem of plasticity loss? (1.2) Do different experts in PA-MoE contribute to alleviating plasticity loss? (1.3) Is the selection probability of each expert in PA-MoE balanced? (2.1) Does PA-MoE offer performance improvements compared to the vanilla MoE architecture?  (2.2) How do individual components of QoE contribute to the overall performance gains observed in PA-MoE? (2.3) How does PA-MoE influence internal system states to ultimately affect performance? (2.4) Does PA-MoE provide performance improvements over both exit learning–based and non–learning-based methods? (3.1) What is the impact of the injected noise scale on the final performance of PA-MoE? (3.2) Does the problem of plasticity loss occur under different QoE coefficients (See Appendix \ref{diverse_beta}.)?

\subsection{Experiment Setting}

The experiments were performed on a system equipped with an Intel(R) Core(TM) i5-10400 CPU @ 2.90GHz, without GPU acceleration. The set of available bitrates is defined as $\mathcal{A} = \{300, 750, 1200, 1850, 2850, 4300\}$ $\mathrm{kbps}$. Each video segment has a duration of 4 seconds. The playback buffer can hold up to 60 seconds of content, and the video consists of 49 segments in total. PPO hyperparameters are detailed in Table~\ref{tab:hyperparams}. For each training run, the agent is trained for approximately two hours, with a total of 2 million timesteps and 1000 iterations. All hyperparameters, including random seeds, are kept consistent across different algorithms. All QoE shift patterns in this paper follow the format shown in Figure~\ref{changing_qoe_action}.


\begin{table}[ht]
\centering
\caption{Potential Hyperparameter Configurations For PPO.}
\label{tab:hyperparams}
\begin{tabular}{ll}
\toprule
\textbf{Hyperparameter} & \textbf{Value} \\
\midrule

Learning Rate & 1e-4 \\
Batch size & 2000 \\
Minibatch Size & 62 \\
Number of Iteration & 1000 \\
Rollout steps per iteration & 2000 \\
Total Timesteps & 2e6 \\
Update Epochs & 5 \\

GAE-$\gamma$ & 0.99 \\
GAE-$\lambda$ & 0.95 \\
Clip $\epsilon$ & 0.2 \\
Entropy Coefficient & 0 \\
Value Function Coefficient & 5 \\
Activation & ReLU \\
Environment & \{D, L, N\} \\
\# Experts & 3 \\
Expert Hidden Size & 18 \\
MoE & \{MoE, SMoE, PA-MoE\} \\
Router & \{Top-K-Router, Softmax-Router\} \\
Number of Selected Experts & 1 \\
Actor MoE & True \\
Critic MoE & True \\
\bottomrule
\end{tabular}
\end{table}

\textbf{Network Trace Datasets:} In our experiments, we draw on three distinct sources of throughput traces: (i) recordings from HSDPA-based 3G networks \cite{riiser2013commute}, collected while smartphones streamed video during travel on subways, trams, trains, buses, and ferries; (ii) the FCC corpus \cite{kan2022improving}, created by stitching together randomly sampled logs from the “Web browsing” class in the August 2016 public release; and (iii) the Puffer open dataset \cite{yan2020learning}, which comprises on-demand video sessions observed over heterogeneous access technologies, including wired links, Wi-Fi, and cellular networks (3G/4G/5G). The training set contains 127 traces and the test set contains 142; there is no overlap between them. These datasets will be released together with the source code upon publication.




\subsection{Does PA-MoE alleviate the issue of plasticity loss?}

A neuron that does not contribute to subsequent layers is known as a Dormant Neuron (for a detailed definition, see Appendix \ref{Expert_Contributions_to_Mitigating_Plasticity_Loss_in_PA_MoE}). The dormant neuron ratio is a widely adopted metric for evaluating neural network plasticity \cite{dohare2024loss, sokar2023dormant, WilliOFDC24}. Figure~\ref{dormant_neuron_rate_evolution} illustrates the evolution of dormant neuron ratios under QoE shifts across various network architectures. Our proposed PA-MoE consistently maintains a lower dormant neuron ratio compared to both Sparse MoE (SMoE) and conventional MoE, and this ratio does not exhibit an increasing trend over time.

\begin{figure}[!htbp]
\centering
\includegraphics[width=2.4in]{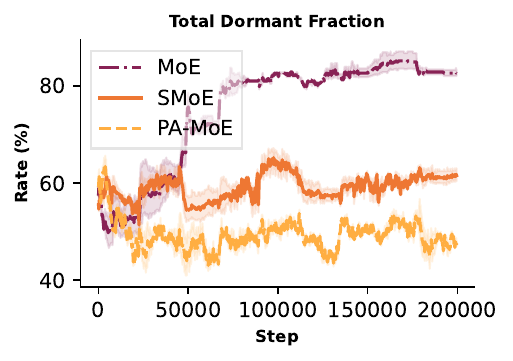}%
\caption{Evolution of neural network dormant neuron rate.}
\label{dormant_neuron_rate_evolution}
\end{figure}

To determine if different experts in PA-MoE contribute to mitigating plasticity loss, we visualize the dormant neuron ratios for each layer in the policy and value networks of every expert (see Appendix \ref{Expert_Contributions_to_Mitigating_Plasticity_Loss_in_PA_MoE}). These visualizations demonstrate that PA-MoE consistently exhibits lower dormant neuron ratios across various layers in both the policy and value networks when compared to conventional MoE and SMoE.

Another widely used metric for characterizing the loss of learning capacity in neural networks is Feature Rank \cite{dohare2024loss, LyleZNPPD23}. In Appendices \ref{rank}, \ref{effective_rank}, \ref{approx_rank}, we show that our PA-MoE method achieves higher effective rank \cite{7098875}, approximate rank, and absolute approximate rank \cite{Yang2020Harnessing} compared to other approaches. These results indicate that PA-MoE preserves a stronger learning capacity with significantly reduced plasticity loss.

\subsection{Are the selection probabilities for each expert balanced?}

\begin{figure*}[ht]
\centering
\includegraphics[width=2.1in]{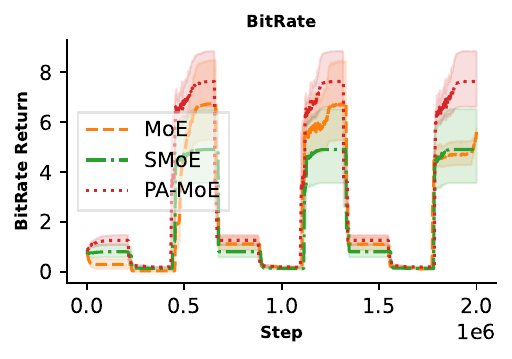}
\includegraphics[width=2.1in]{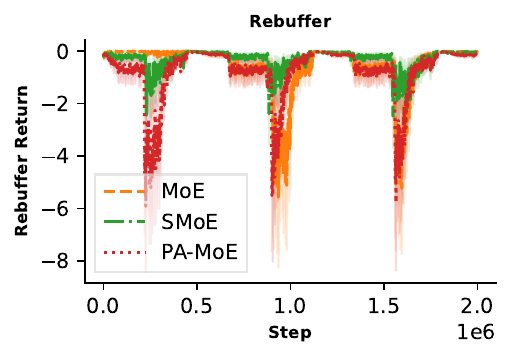}
\includegraphics[width=2.1in]{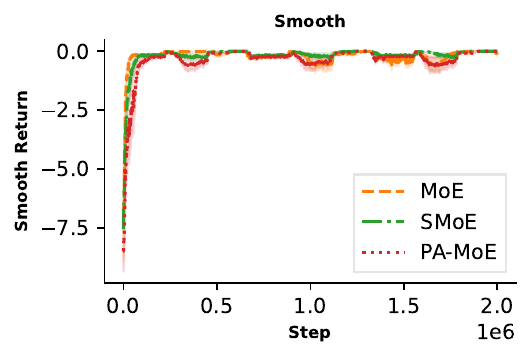}
\caption{Performance Comparison of Different Algorithms Based on QoE Component Metrics}
\label{ppo_seperate_reward}
\end{figure*}

\begin{figure}[!htbp]
\centering
\includegraphics[width=\linewidth]{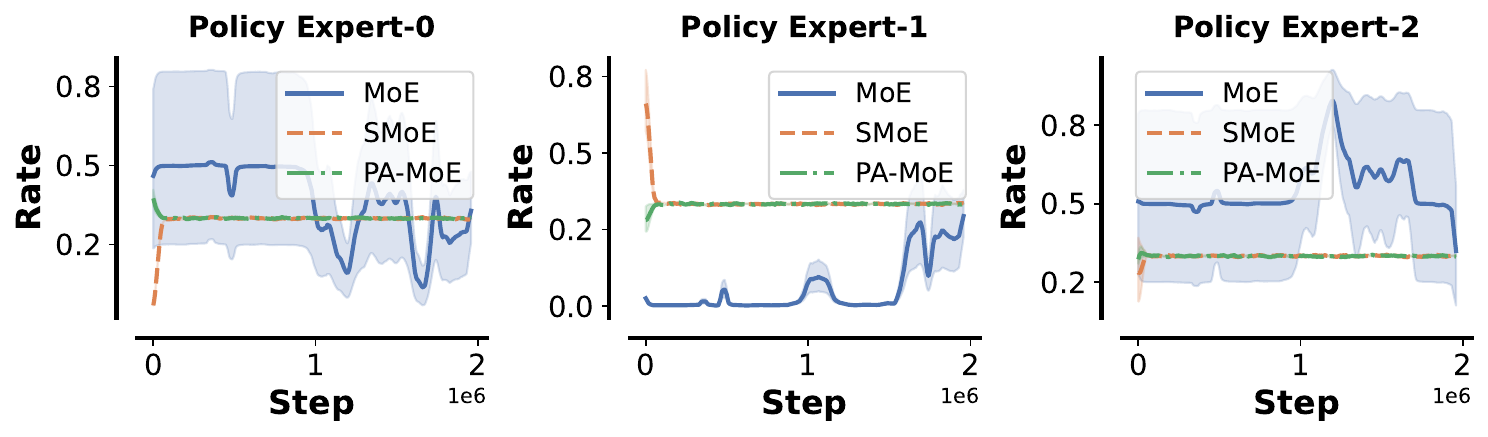}
\includegraphics[width=\linewidth]{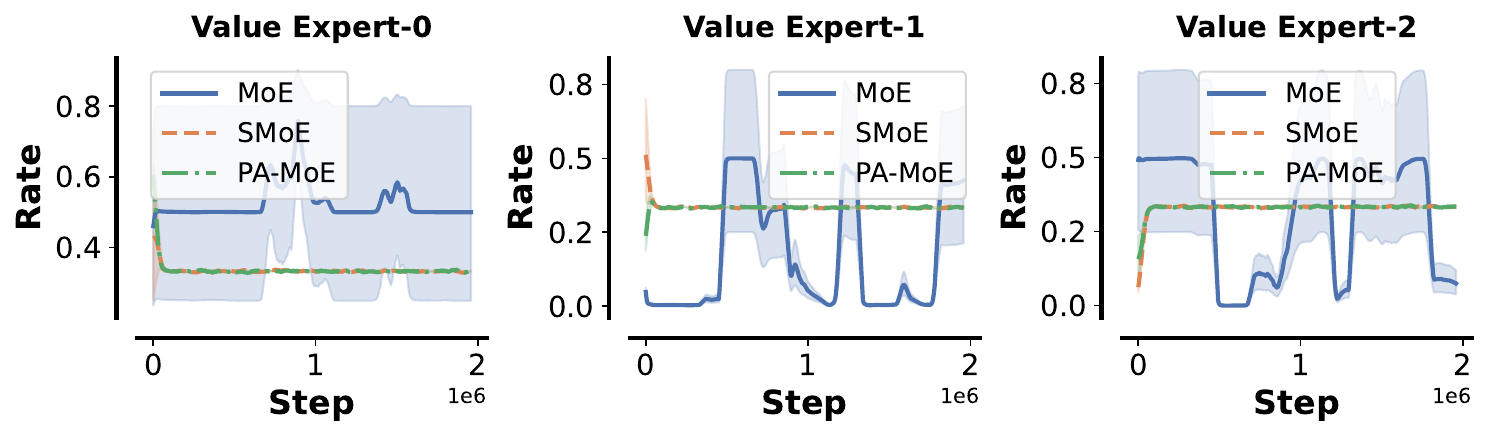}
\caption{Probability distribution of each expert selected for Policy and Value.}
\label{load_balanced_visual}
\end{figure}

Load balancing is a critical factor in evaluating the effectiveness of MoE designs. In a load-balanced MoE, all experts are utilized evenly, preventing some experts from being overused while others remain idle. This leads to improved computational efficiency. Additionally, a balanced allocation promotes diverse feature learning among experts, ultimately enhancing the model’s ability to handle complex tasks. To validate the load-balancing capability of our proposed PA-MoE, we visualized the selection probability of each expert in both the policy and value functions. As shown in Figure~\ref{load_balanced_visual}, our experimental results clearly demonstrate that each expert in PA-MoE is selected with equal probability. This balanced selection effectively avoids over-reliance on a subset of experts, thereby reducing the risk of overfitting and ensuring robust performance across diverse tasks. By dynamically routing, the computational load on neurons can be distributed, thereby alleviating the loss of plasticity.

\subsection{Does PA-MoE Yield Performance Improvements?}

Another question that we are interested in is: Does PA-MoE result in performance improvements? Additionally, we would like to understand how the individual components of QoE contribute to the overall QoE under Shift QoE conditions. The variation curve of the overall QoE in a Shift QoE environment is depicted in Figure~\ref{qoe_reward}.

\begin{figure}[ht]
\centering
\includegraphics[width=2.1in]{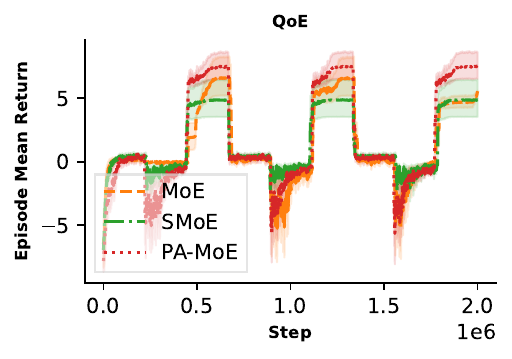}
\caption{Final QoE Performance Comparison of Different Algorithms under Shift QoE.}
\label{qoe_reward}
\end{figure}

As illustrated in Figure~\ref{qoe_reward}, it is clear that MoE requires a great number of steps to generate enough interaction data to adapt to the changing QoE, while PA-MoE quickly adapts and achieves higher performance gains. To more accurately measure these improvements, we utilize the more reliable reinforcement learning metric, the IQM (Interquartile Mean) \cite{agarwal2021deep}, as shown in Figure~\ref{decomposition_reward_iqm}. This method aims to emphasize the relative performance of different approaches while reducing the influence of outliers.

\begin{figure}[!htbp]
\centering
\includegraphics[width=2.1in]{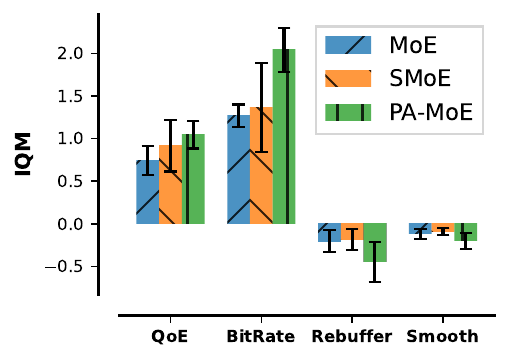}%
\caption{Comparison of QoE Methods using the IQM with a focus on the 25th-75th percentile returns.}
\label{decomposition_reward_iqm}
\end{figure}

The results shown in Figure~\ref{decomposition_reward_iqm} demonstrate that our PA-MoE method improves overall QoE by 45.5\% compared to conventional MoE. The most significant improvements are seen in the Bitrate component, indicating that the agent is able to selectively optimize the different QoE components to meet the requirements of a Shift QoE environment. More detailed changes for each metric can be seen in Figure~\ref{ppo_seperate_reward}.

\subsection{PA-MOE vs. Existing ABR algorithms}

\begin{figure*}[!htbp]
\centering
\includegraphics[width=1.8in]{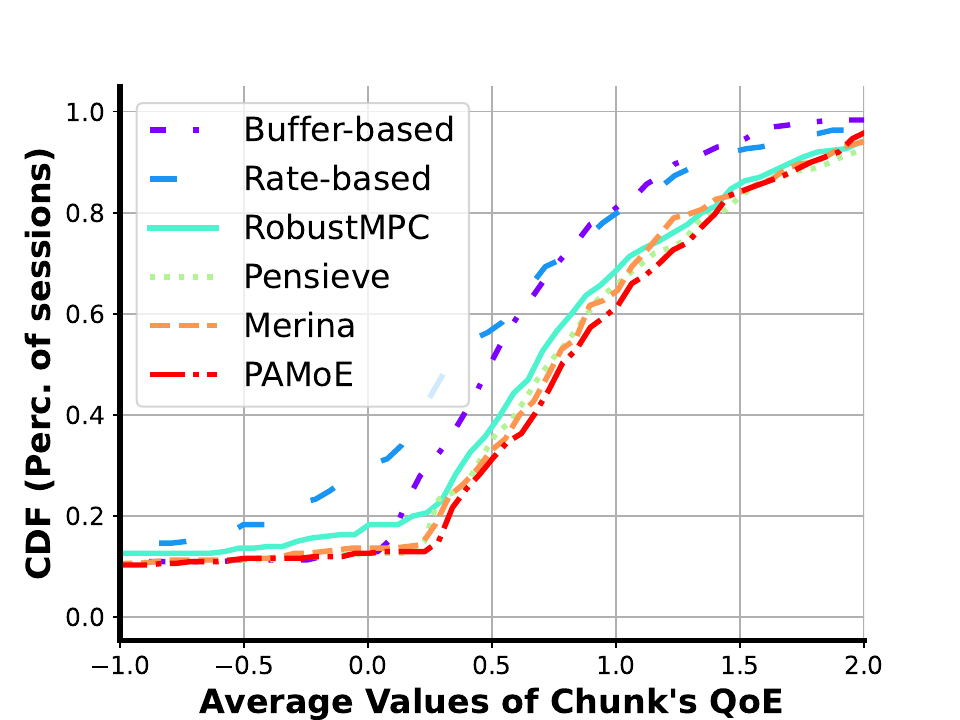}%
\includegraphics[width=1.8in]{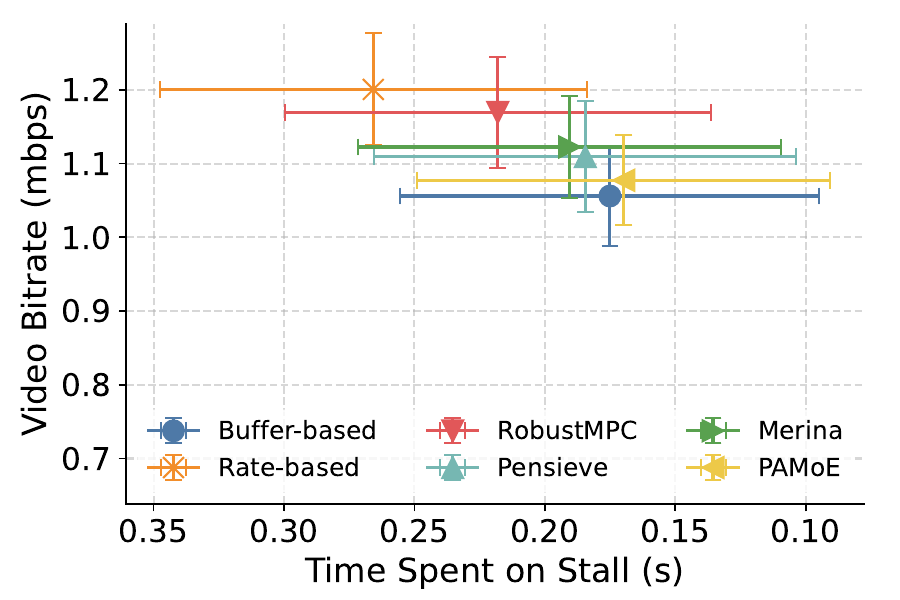}%
\includegraphics[width=1.8in]{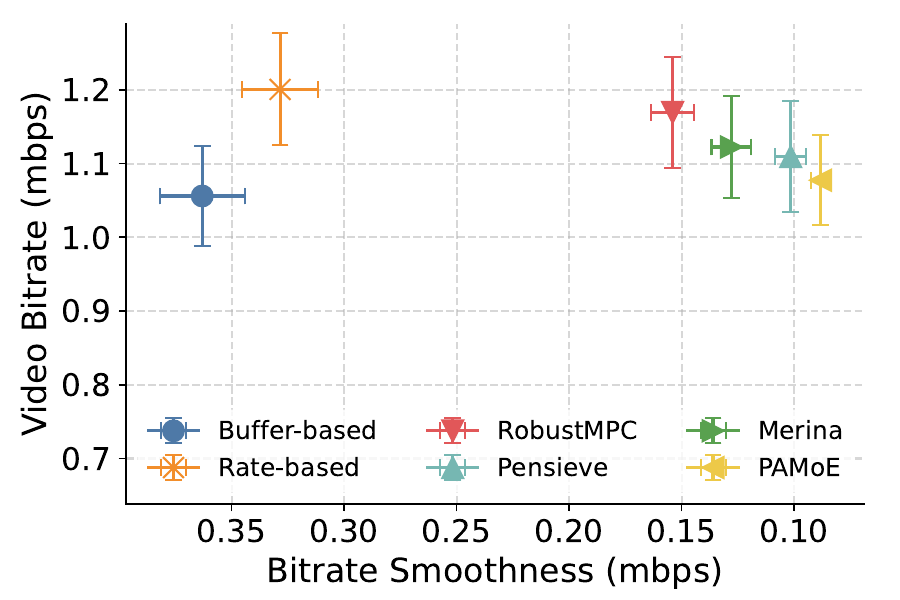}%
\includegraphics[width=1.8in]{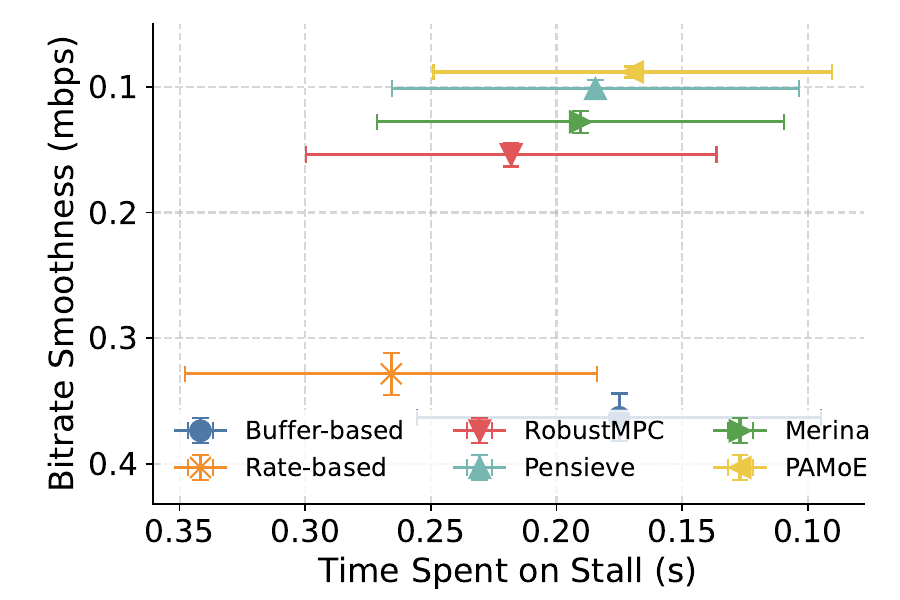}%
\caption{Comparing PA-MOE with recent ABR algorithms over the Train set.}
\label{pamoe_train_reward}
\end{figure*}

\begin{figure*}[!htbp]
\centering
\includegraphics[width=1.8in]{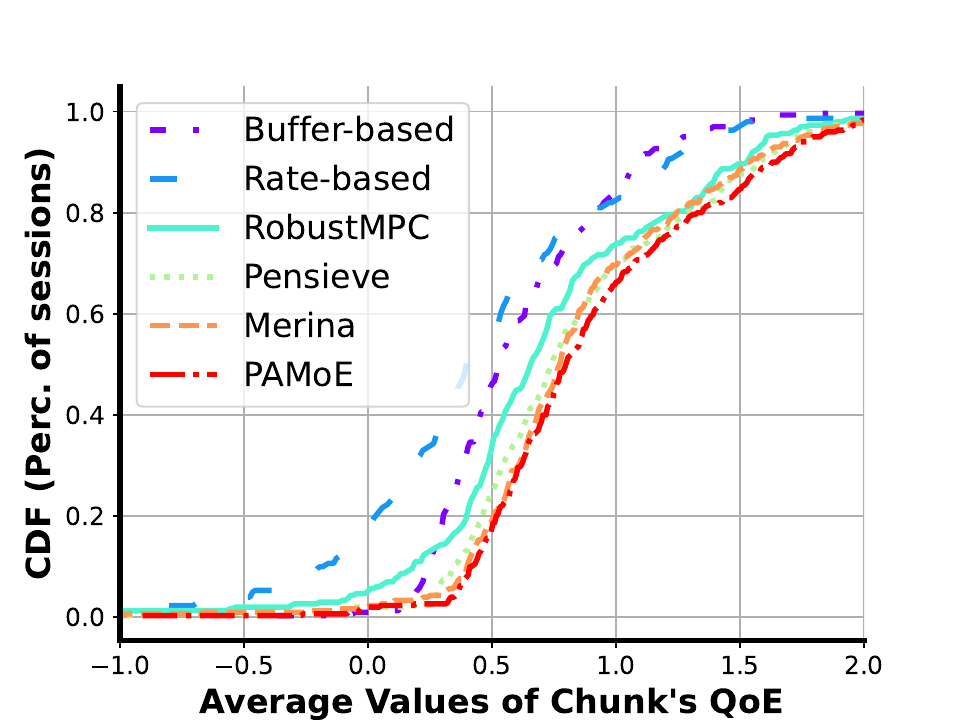}%
\includegraphics[width=1.8in]{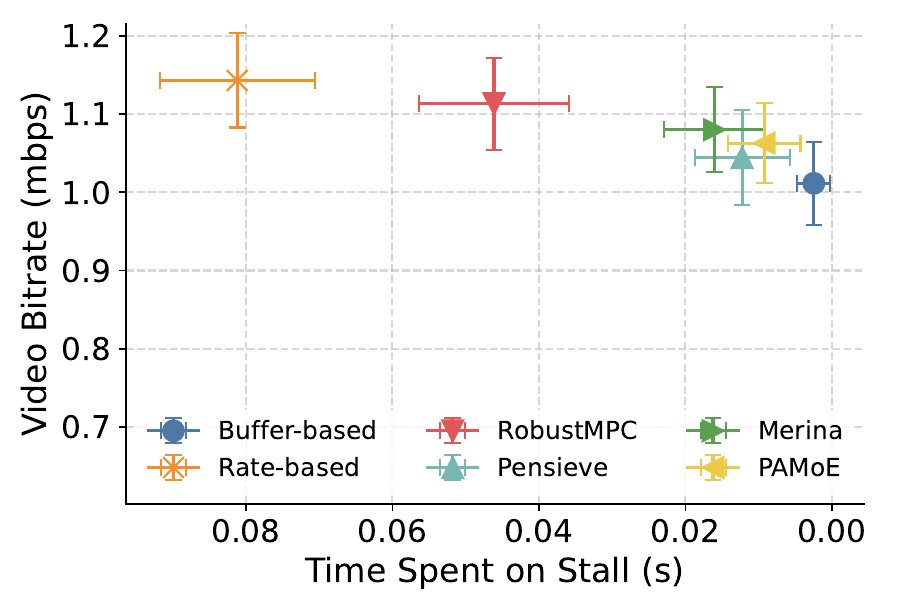}%
\includegraphics[width=1.8in]{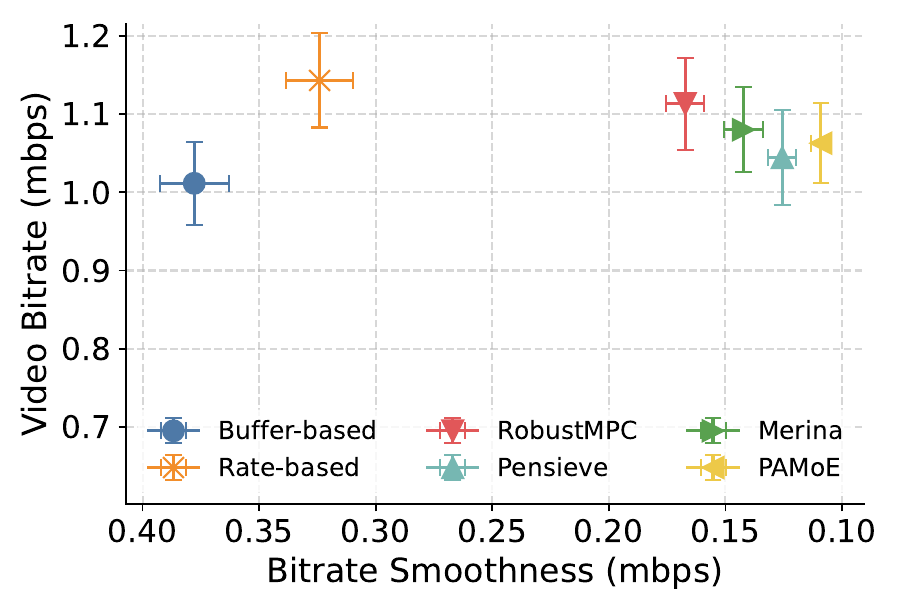}%
\includegraphics[width=1.8in]{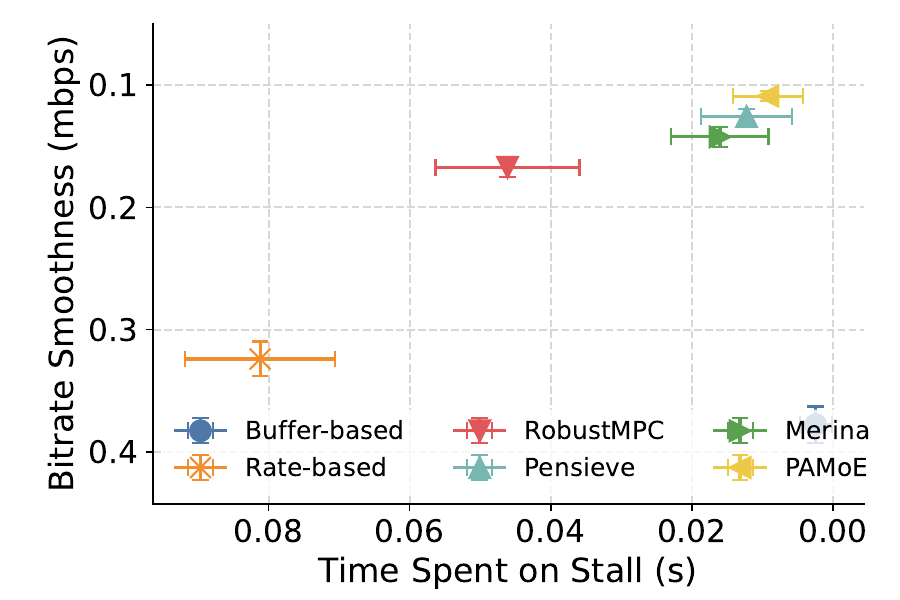}%
\caption{Comparing PA-MOE with recent ABR algorithms over the Test set.}
\label{pamoe_test_reward}
\end{figure*}

In this subsection, we compare our approach with methods that rely on prior knowledge, including learning-based approaches such as Pensieve \cite{mao2017neural} and the meta-learning method Merina \cite{kan2022improving}, as well as non-learning-based approaches such as RobustMPC, RateBased, and BufferBased \cite{yin2015control}. To ensure a fair and credible comparison, we use the same QoE-component coefficients and the same network architecture as in Merina \cite{kan2025merina+} for each expert in the MoE. Similar to Pensieve \cite{mao2017neural} and the meta-learning method Merina \cite{kan2022improving}, we use the final trained model for performance evaluation. The network bandwidth is randomly sampled from either the training traces or the test traces, and each video download constitutes one episode.

For each algorithm, we run 300 episodes under identical settings and compute the average QoE for each episode; these per-episode averages are then used for comparing QoE performance across algorithms. The cumulative distribution functions (CDFs) plot the empirical cumulative distribution of QoE values across the 300 episodes. We aggregate the per-episode QoE values, compute a histogram using 500 equally spaced bins over the QoE range, take the cumulative sum of bin counts, and normalize by the total number of episodes to obtain the empirical CDF. Thus, at any position on the horizontal axis, the curve reports the fraction of sessions whose metric value is less than or equal to the corresponding horizontal value. No additional smoothing or parametric assumptions are used; the curves directly reflect the empirical distribution of the metric for each scheme. We perform pairwise comparisons of the QoE components—Bitrate Reward, Stall Time, and Smoothness—and compute the sample means together with 95\% confidence intervals based on the Student’s $t$-distribution.

\begin{table*}[t]
  \begingroup
  \centering
  \caption{QoE and its components for different ABR algorithms on test dataset (mean over episodes).}
  \label{tab:result}
  \begin{tabular}{lcccc}
    \toprule
    Algorithm &
    Mean QoE $\uparrow$ &
    Bitrate (Mbps) $\uparrow$ &
    Rebuffer Time (s) $\downarrow$ &
    Smoothness (Mbps) $\downarrow$ \\
    \midrule

    \multicolumn{5}{l}{\textbf{Rule-based methods}} \\
    Buffer Based & $0.623$ & $1.011 \pm 0.053$ & $0.002 \pm 0.002$ & $0.378 \pm 0.015$ \\
    Rate Based   & $0.470$ & $1.143 \pm 0.060$ & $0.081 \pm 0.011$ & $0.324 \pm 0.014$ \\
    Robust MPC   & $0.748$ & $1.113 \pm 0.059$ & $0.046 \pm 0.010$ & $0.167 \pm 0.008$ \\
    \addlinespace[0.4em]

    \multicolumn{5}{l}{\textbf{Learning-based methods}} \\
    Pensieve & $0.866$ & $1.044 \pm 0.060$ & $0.012 \pm 0.007$ & $0.126 \pm 0.006$ \\
    Merina   & $0.869$ & $1.080 \pm 0.055$ & $0.016 \pm 0.007$ & $0.142 \pm 0.084$ \\
    PA-MoE   & $\textbf{0.914}$ & $1.063 \pm 0.051$ & $0.009 \pm 0.005$ & $0.109 \pm 0.004$ \\
    \bottomrule
  \end{tabular}

  \vspace{0.3em}
  {\footnotesize
  \textit{Note:} $\uparrow$ indicates larger is better; $\downarrow$ indicates smaller is better.
  Bitrate, Rebuffer Time, and Smoothness are reported as mean $\pm$ 95\% Confidence Interval (CI);
  Mean QoE is reported as mean only.}
  \endgroup
\end{table*}

Figure \ref{pamoe_train_reward} presents the CDFs of average QoE for all sessions and algorithms on the training set, along with pairwise comparisons of the QoE components—bitrate, smoothness, and stall time. Figure \ref{pamoe_test_reward} shows the corresponding results on the testing set. Table~\ref{tab:result} reports the QoE and its components achieved by different ABR algorithms on the test set. As shown, our proposed PA-MoE achieves state-of-the-art performance, even relative to the meta-learning method Merina, which leverages prior knowledge. These results underscore the substantial potential of optimizing adaptive bitrate (ABR) algorithms through the lens of plasticity.

\subsection{Influence on System Internal States}

\begin{figure}[!htbp]
\centering
\includegraphics[width=2.2in]{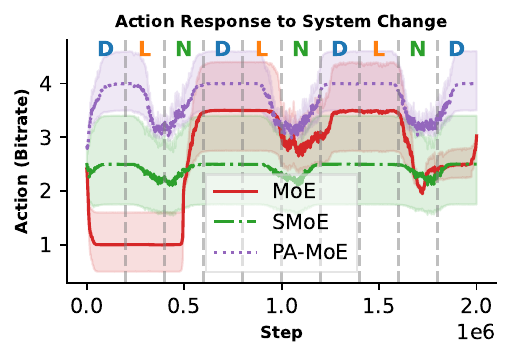}%
\caption{Illustration of Internal State Metrics Variations in a Video Streaming System under Shift QoE Conditions.}
\label{decomposition_reward_system}
\end{figure}

In this section, we present visualizations of the system's internal state metrics to analyze the behavior of different algorithms. This will help us gain a better understanding of their impact on overall QoE. In this section, we present visualizations of the system's internal state metrics to analyze the behavior of different algorithms. This will help us gain a better understanding of their impact on overall QoE. The system environment changes in Figure~\ref{decomposition_reward_system} are consistent with those in Figure~\ref{changing_qoe_action}. ‘D’ denotes documentaries, ‘L’ live streams, and ‘N’ news. The parameter coefficients for each video type are also the same as in Figure~\ref{changing_qoe_action}. Because the objective function changes abruptly, we expect the algorithm to adjust its behavior when the QoE coefficients vary. For example, when the objective function shifts from $L$ to $N$, meaning that the bitrate coefficient in the QoE changes from 1 to $\beta$ while the rebuffering coefficient changes from $\beta$ to 1, the agent should correspondingly modify its action—specifically, by increasing the selected bitrate—to obtain a higher QoE reward.

\begin{figure}[!htbp]
\centering
\includegraphics[width=3.1in]{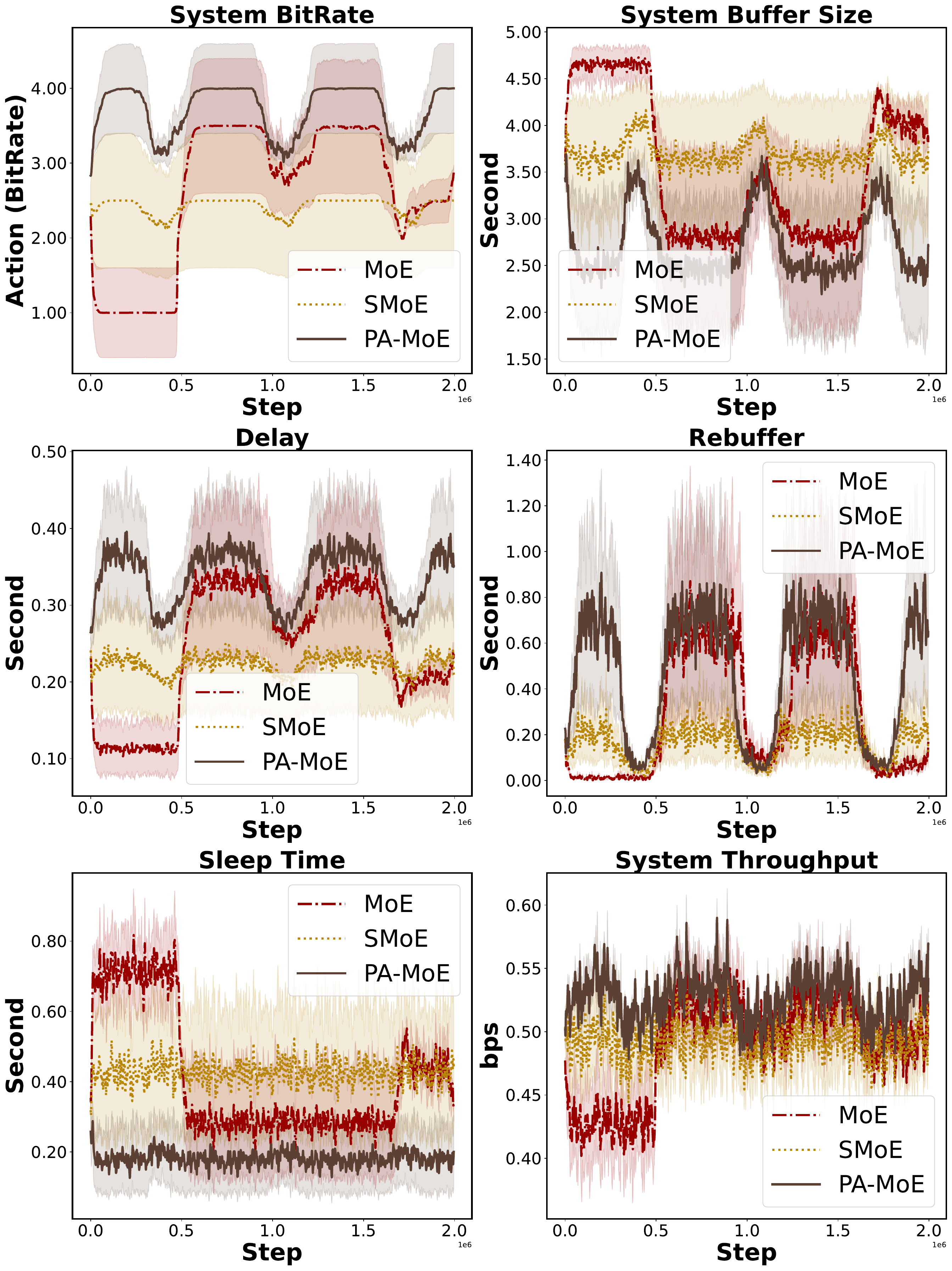}%
\caption{Variations in the Internal State Metrics of the Video Streaming System.}
\label{decomposition_reward_system_all}
\end{figure}

As shown in Figure~\ref{decomposition_reward_system}, we can see that PA-MoE exhibits periodic variations in the action space, which correspond to system changes. On the other hand, SMoE only shows minor fluctuations and fails to adapt quickly. In contrast, MoE does not demonstrate any periodic adaptation. Since the QoE coefficients change directly, using the QoE value itself to study plasticity is not appropriate. If the policy does not adapt its behavior in response to environmental changes, the direct modification of QoE coefficients will still cause the QoE value to fluctuate accordingly. Therefore, examining whether the actions exhibit periodic adjustments provides a more reliable reflection of the system’s adaptive state.

Furthermore, for an AVS system, when the coefficient corresponding to the bitrate term in the QoE increases, the algorithm should select a higher bitrate to achieve a better QoE score. However, a higher bitrate requires more time to download, leading to increased delay and rebuffering time, while the buffer size decreases and the throughput increases. This represents a preliminary analysis of the internal state dynamics of the system.

Figure~\ref{decomposition_reward_system_all} illustrates detailed variations across the entire internal state of the system, revealing that PA-MoE more effectively utilizes the buffer size, resulting in higher throughput. From the perspective of internal system state variables, the performance improvement of PA-MOE mainly comes from reduced sleep time and more efficient buffer utilization. By adapting to changes in the action space, it can actively select actions that yield higher rewards. In contrast, MoE suffers from a loss of plasticity and fails to keep up with the rapidly changing system. This further validates the effectiveness of our proposed plasticity injection method.

\subsection{Parameter Sensitivity Analysis in PA-MoE}

\begin{figure}[!htbp]
\centering
\includegraphics[width=2.4in]{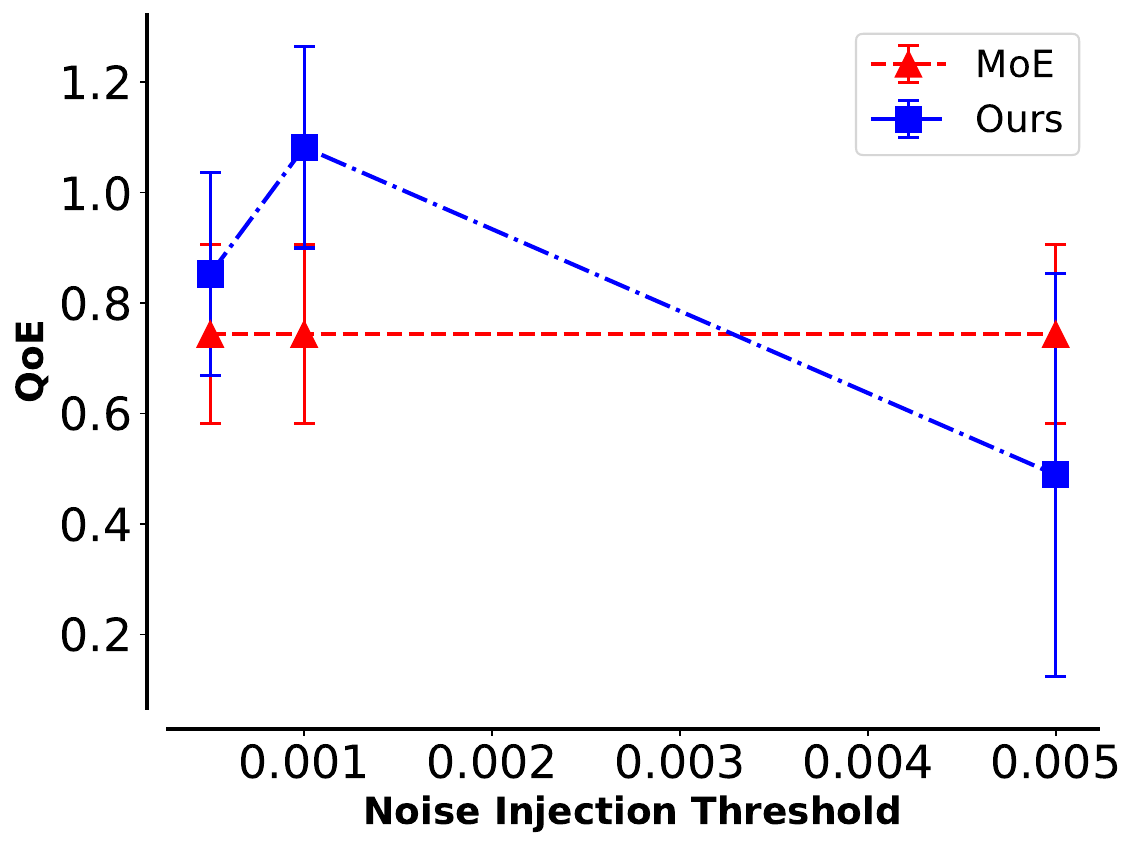}%
\caption{Comparison of QoE Methods using the IQM  (Interquartile Mean) with a focus on the 25th-75th percentile returns. This approach aims to highlight the relative performance of different methods while minimizing the impact of outliers.}
\label{noise}
\end{figure}

In PA-MoE, the magnitude of noise injection is a critical parameter. Our sensitivity analysis, presented in Fig. \ref{noise}, reveals that both excessively low and high noise levels lead to performance degradation. This observation aligns with our theoretical analysis and validates Theorem \ref{thm:error-bound}. Furthermore, as shown in Appendix \ref{Expert_weight_to_Mitigating_Plasticity_Loss_in_PA_MoE}, visualizations of the weight magnitudes indicate that the neural network adaptively adjusts its weights to retain memory and counteract noise-induced forgetting. These findings underscore the importance of carefully tuning the noise level to balance plasticity and stability in dynamic environments.

\begin{figure}[!htbp]
\centering
\includegraphics[width=2.4in]{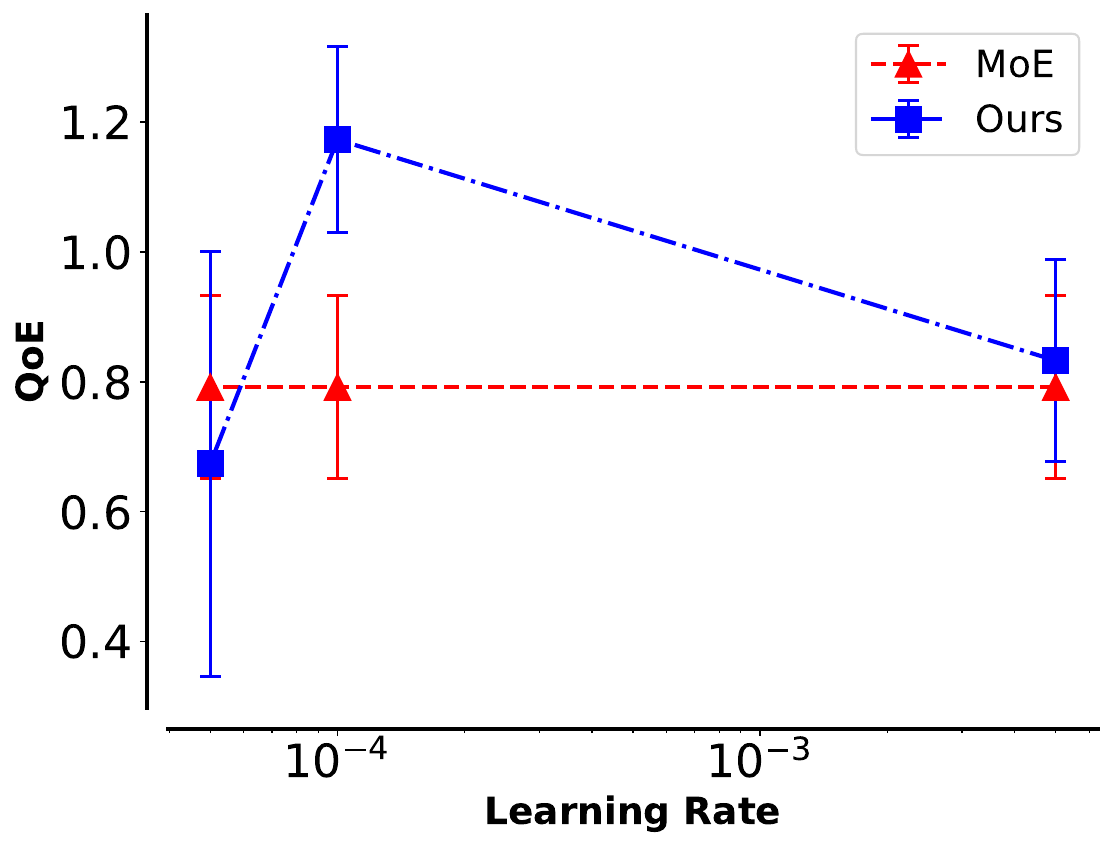}%
\caption{Comparison of QoE Methods using the IQM  (Interquartile Mean) with a focus on the 25th-75th percentile returns.}
\label{lr}
\end{figure}

Figure~\ref{lr} shows QoE as a function of the learning rate. We find that both excessively large and excessively small learning rates degrade performance, in agreement with our theoretical results. An inappropriate learning rate may even violate the trust-region assumption. Nevertheless, the experiments indicate that, even with a learning rate as large as 0.005—rare in practice—our method still achieves a modest improvement.

\section{Conclusion}

In this paper, we revisit the paradigm of optimization under changing objective functions in adaptive video streaming. We introduce PA-MoE, a novel approach designed to mitigate neural network plasticity loss during the optimization process under dynamic QoE shift scenarios. Our method leverages controlled noise injection to enable robust and rapid adaptation. We provide rigorous theoretical insights and analysis to support the effectiveness of PA-MoE, and extensive experiments validate both the identified problem and our proposed solution from multiple perspectives. Additionally, parameter sensitivity studies further corroborate our theoretical claims. Future work may explore more advanced strategies for alleviating plasticity loss, building on the foundational noise injection mechanism presented here. Another promising direction is the development of methods for automatically adjusting the injected noise to optimize learning stability and performance.

\bibliographystyle{IEEEtran}
\bibliography{reference}

\begin{IEEEbiography}[{\includegraphics
[width=1in,height=1.25in,clip,
keepaspectratio]{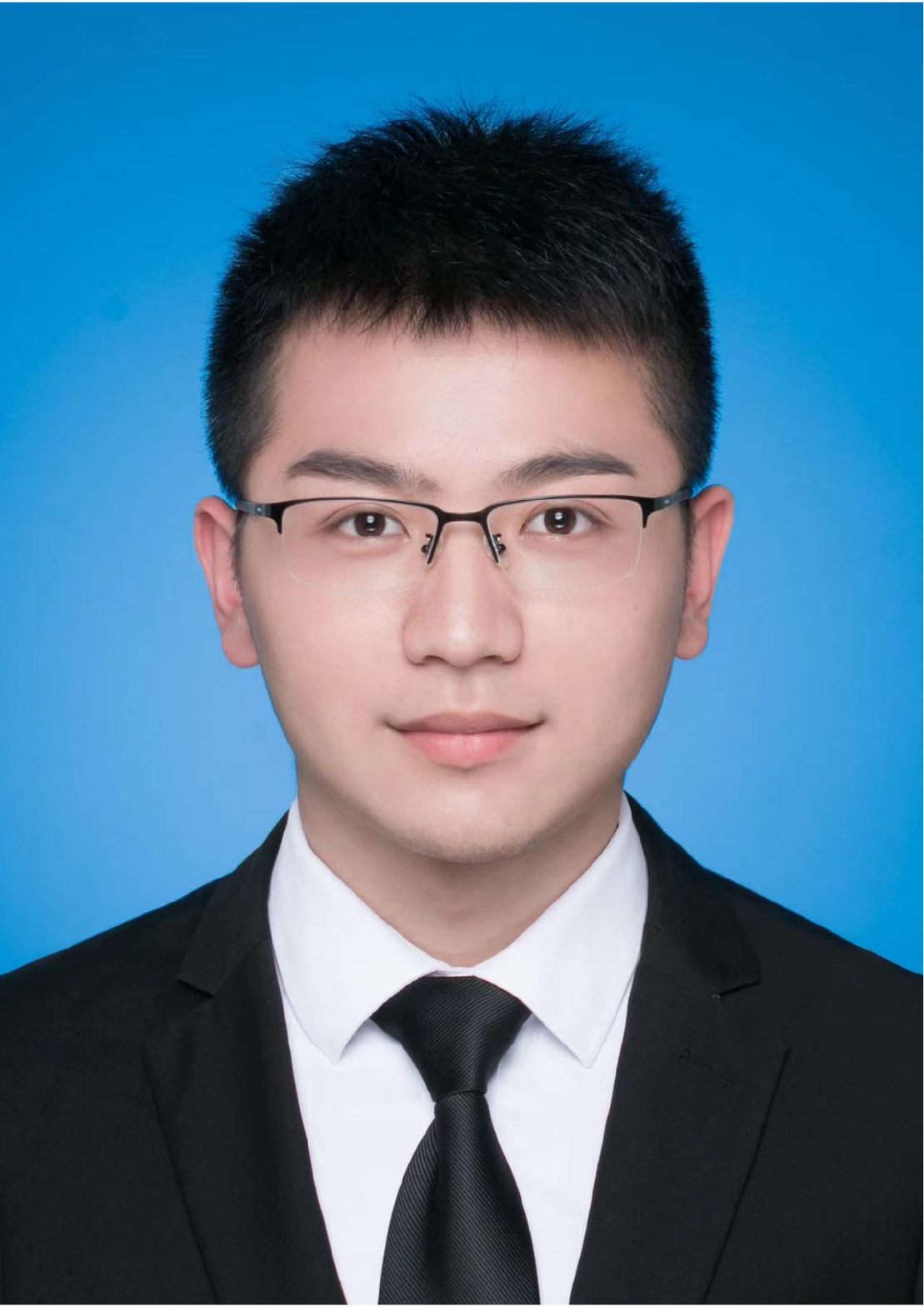}}]
{Zhiqiang He} is currently a Ph.D at the University of Electro-Communications. He received his MS degree in Control Science and Engineering from Northeastern University, Shenyang, China. His research interests focus on deep reinforcement learning and its control applications. He previously worked at Baidu and InspirAI, where he developed a master-level AI for the game Landlord that outperformed professional players.
\end{IEEEbiography}

\vspace{5pt}

\begin{IEEEbiography}
[{\includegraphics[width=1in,height=1.25in,clip,keepaspectratio]{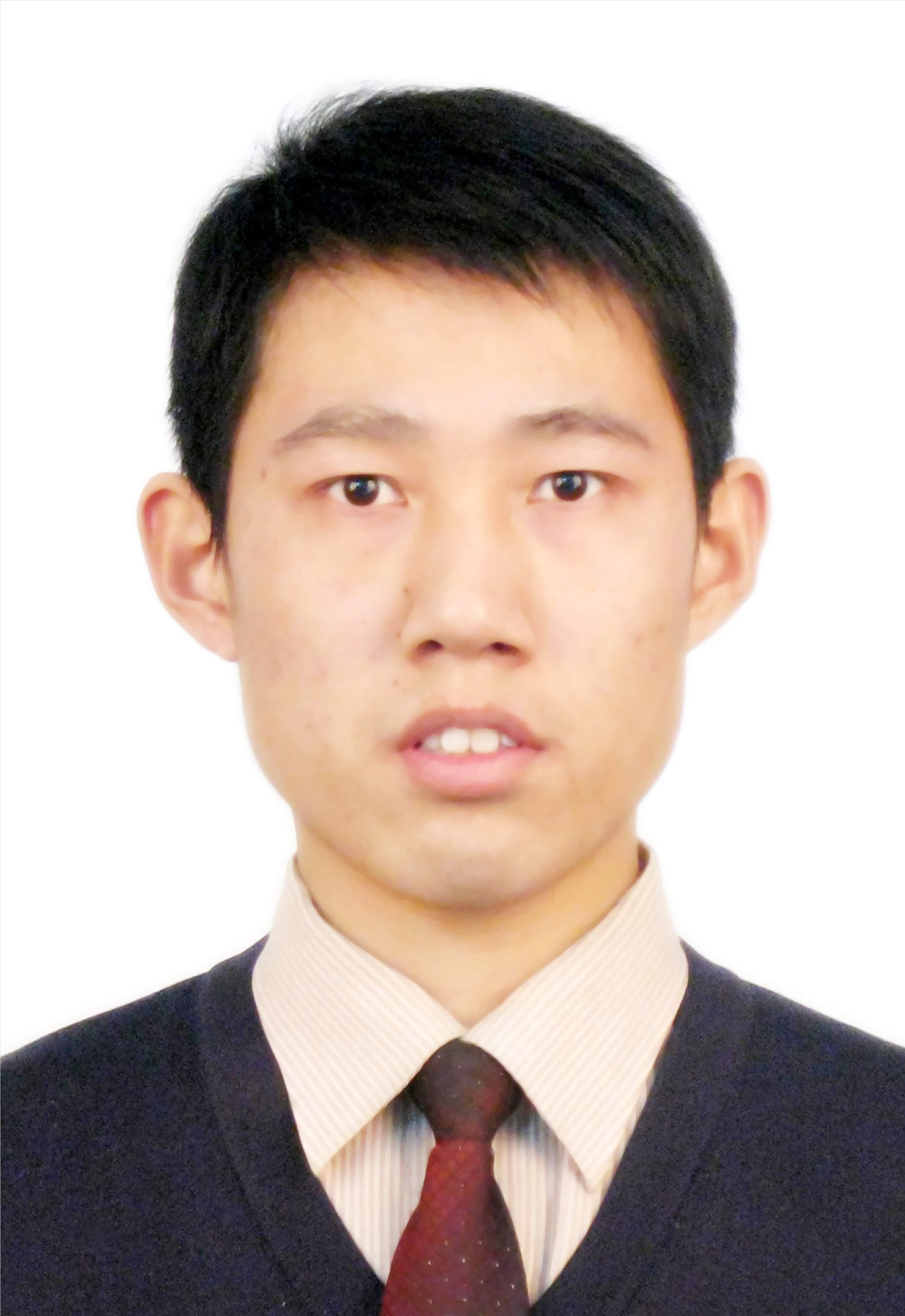}}]
{Zhi Liu} (S’11-M’14-SM’19) received the Ph.D. degree in informatics in National Institute of Informatics. He is currently an Associate Professor at the University of Electro-Communications. His research interest includes video network transmission and MEC. He is now an editorial board member of IEEE Transactions on Multimedia, IEEE Networks and Internet of Things Journal. He is a senior member of IEEE.
\end{IEEEbiography}

\vfill

\newpage
\onecolumn
\newpage
\twocolumn
\appendix

\subsection{Proof of Gradient Co-coercivity Lemma Proof.}\label{appendix_proof_gradient_co_coercivity_lemma}

\textbf{Lemma \ref{GradientCocoercivityLemma}:}  (\textbf{Gradient Co-coercivity Lemma.})
Let $L_t:\mathbb{R}^d\to\mathbb{R}$ be an $L$–smooth function; that is, for all $\boldsymbol{w},\boldsymbol{w}'\in\mathbb{R}^d$, $\|\nabla L_t(\boldsymbol{w})-\nabla L_t(\boldsymbol{w}')\| \le L\|\boldsymbol{w}-\boldsymbol{w}'\|$. Then, for any $\boldsymbol{w}_i^t,\boldsymbol{w}_t^*\in\mathbb{R}^d$, we have
\begin{align}\label{eq:main-ineq_appendix}
\|\nabla L_t(\boldsymbol{w}_i^t) & -\nabla L_t(\boldsymbol{w}_t^*)\|^2  \\ 
& \le L\, (\nabla L_t(\boldsymbol{w}_i^t)-\nabla L_t(\boldsymbol{w}_t^*))^{T} (\boldsymbol{w}_i^t-\boldsymbol{w}_t^*). \nonumber
\end{align}

\begin{proof}
Since $L_t$ is $L$–smooth, we immediately have $\|\nabla L_t(\boldsymbol{w}_i^t)-\nabla L_t(\boldsymbol{w}_t^*)\| \le L\|\boldsymbol{w}_i^t-\boldsymbol{w}_t^*\|$. Multiplying both sides by $\|\boldsymbol{w}_i^t-\boldsymbol{w}_t^*\|$ yields

\begin{align}
    0 & \leq L \left \| \boldsymbol{w}_{i}^{t} - \boldsymbol{w}_{t}^{*} \right \|^{2} -  \left \| \nabla L_{t}(\boldsymbol{w}_{i}^{t}) - \nabla L_{t}(\boldsymbol{w}_{t}^{*}) \right \|  \left \| \boldsymbol{w}_{i}^{t} - \boldsymbol{w}_{t}^{*} \right \| \nonumber \\
    & \leq L \left \| \boldsymbol{w}_{i}^{t} - \boldsymbol{w}_{t}^{*} \right \|^{2} - (\nabla L_{t}(\boldsymbol{w}_{i}^{t}) - \nabla L_{t}(\boldsymbol{w}_{t}^{*}))^{T}(\boldsymbol{w}_{i}^{t} - \boldsymbol{w}_{t}^{*}) \nonumber \\
    & \leq (L \boldsymbol{w}_{i}^{t} - \nabla L_{t}(\boldsymbol{w}_{i}^{t})- L \boldsymbol{w}_{t}^{*} + \nabla L_{t}(\boldsymbol{w}_{t}^{*}))(\boldsymbol{w}_{i}^{t} - \boldsymbol{w}_{t}^{*}). \nonumber
\end{align}

Next, define the auxiliary function $g(\boldsymbol{w}) := \frac{1}{2} L \boldsymbol{w}^{T} \boldsymbol{w} - L_{t}(\boldsymbol{w})$, so that $\nabla g(\boldsymbol{w}) = L\,\boldsymbol{w} - \nabla L_t(\boldsymbol{w})$. Because $L_t$ is $L$–smooth, one may verify that $g(\boldsymbol{w})$ is convex. Hence, for any $\boldsymbol{w}_i^t$ and $\boldsymbol{w}_t^*$ we have $g(\boldsymbol{w}_i^t) \ge g(\boldsymbol{w}_t^*) + \nabla g(\boldsymbol{w}_t^*)^T (\boldsymbol{w}_i^t-\boldsymbol{w}_t^*)$, that is,
\begin{align}\label{eq:convex-g}
\frac{L}{2}&\|\boldsymbol{w}_i^t\|^2 - L_t(\boldsymbol{w}_i^t) \\ 
& \ge \frac{L}{2}\|\boldsymbol{w}_t^*\|^2 - L_t(\boldsymbol{w}_t^*) + \left(L\boldsymbol{w}_t^* - \nabla L_t(\boldsymbol{w}_t^*)\right)^T (\boldsymbol{w}_i^t-\boldsymbol{w}_t^*). \nonumber
\end{align}
Rearranging \eqref{eq:convex-g} gives
\begin{align}\label{eq:diff-1}
\frac{L}{2}\|\boldsymbol{w}_i^t & -\boldsymbol{w}_t^*\|^2 \ge \nonumber \\
& L_t(\boldsymbol{w}_i^t) - L_t(\boldsymbol{w}_t^*) - \nabla L_t(\boldsymbol{w}_t^*)^T (\boldsymbol{w}_i^t-\boldsymbol{w}_t^*).
\end{align}
Now, introduce the point $z = \boldsymbol{w}_t^* + \frac{1}{L}\Bigl(\nabla L_t(\boldsymbol{w}_i^t)-\nabla L_t(\boldsymbol{w}_t^*)\Bigr)$. Using a standard descent argument and the smoothness of $L_t$, one can write
\begin{align}
L_t(\boldsymbol{w}_t^*)-L_t(\boldsymbol{w}_i^t)
&=\Bigl[L_t(\boldsymbol{w}_t^*)-L_t(z)\Bigr] + \Bigl[L_t(z)-L_t(\boldsymbol{w}_i^t)\Bigr] \nonumber\\[1mm]
&\ge -\nabla L_t(\boldsymbol{w}_t^*)^T (z-\boldsymbol{w}_t^*) - \frac{L}{2}\|z-\boldsymbol{w}_t^*\|^2 \nonumber\\[1mm]
&\quad\; + \nabla L_t(\boldsymbol{w}_i^t)^T (z-\boldsymbol{w}_i^t). \label{eq:diff-2}
\end{align}
Observing that $z-\boldsymbol{w}_t^*=\frac{1}{L}\Bigl(\nabla L_t(\boldsymbol{w}_i^t)-\nabla L_t(\boldsymbol{w}_t^*)\Bigr)$ and $z-\boldsymbol{w}_i^t = (\boldsymbol{w}_t^*-\boldsymbol{w}_i^t) + \frac{1}{L}\Bigl(\nabla L_t(\boldsymbol{w}_i^t)-\nabla L_t(\boldsymbol{w}_t^*)\Bigr)$, one obtains after elementary algebra that \eqref{eq:diff-2} can be rearranged to
\begin{align}\label{eq:diff-3}
L_t(\boldsymbol{w}_t^*)&-L_t(\boldsymbol{w}_i^t) \ge \\
&\nabla L_t(\boldsymbol{w}_i^t)^T (\boldsymbol{w}_t^*-\boldsymbol{w}_i^t) + \frac{1}{2L}\|\nabla L_t(\boldsymbol{w}_i^t)-\nabla L_t(\boldsymbol{w}_t^*)\|^2. \nonumber
\end{align}
In other words,
\begin{align}
L_t(\boldsymbol{w}_t^*) &\ge L_t(\boldsymbol{w}_i^t) + \nabla L_t(\boldsymbol{w}_i^t)^T (\boldsymbol{w}_t^*-\boldsymbol{w}_i^t) + \\
& \frac{1}{2L}\|\nabla L_t(\boldsymbol{w}_i^t)-\nabla L_t(\boldsymbol{w}_t^*)\|^2, \nonumber \\[1mm]
L_t(\boldsymbol{w}_i^t) &\ge L_t(\boldsymbol{w}_t^*) + \nabla L_t(\boldsymbol{w}_t^*)^T (\boldsymbol{w}_i^t-\boldsymbol{w}_t^*) + \\
& \frac{1}{2L}\|\nabla L_t(\boldsymbol{w}_i^t)-\nabla L_t(\boldsymbol{w}_t^*)\|^2. \nonumber
\end{align}
Adding the two relations (or, equivalently, combining them appropriately) leads to
\begin{equation}\label{eq:final-step}
L\,\Bigl(\nabla L_t(\boldsymbol{w}_i^t)-\nabla L_t(\boldsymbol{w}_t^*)\Bigr)^T(\boldsymbol{w}_i^t-\boldsymbol{w}_t^*) \ge \|\nabla L_t(\boldsymbol{w}_i^t)-\nabla L_t(\boldsymbol{w}_t^*)\|^2. \nonumber
\end{equation}
This completes the proof of \eqref{eq:main-ineq_appendix}.
\end{proof}

\subsection{Gradient Strong Convexity Lemma Proof.}\label{appendix_proof_strongconvexity_lemma}

\textbf{Lemma \ref{GradientStrongConvexity}:}  (\textbf{Gradient Strong Convexity Lemma.}) Let $L_t:\mathbb{R}^d\to\mathbb{R}$ be a $\mu$–strongly convex function; that is, for all $\boldsymbol{w},\boldsymbol{w}'\in\mathbb{R}^d$, $L_t(\boldsymbol{w}') \ge L_t(\boldsymbol{w}) + \nabla L_t(\boldsymbol{w})^{T} (\boldsymbol{w}'-\boldsymbol{w}) + \frac{\mu}{2}\|\boldsymbol{w}'-\boldsymbol{w}\|^2$. Then, for any $\boldsymbol{w}_i^t,\boldsymbol{w}_t^*\in\mathbb{R}^d$, it holds that
\begin{equation}\label{eq:strong-convexity_appendi}
\left(\nabla L_t(\boldsymbol{w}_i^t) - \nabla L_t(\boldsymbol{w}_t^*)\right)^{T} (\boldsymbol{w}_i^t-\boldsymbol{w}_t^*) \ge \mu \|\boldsymbol{w}_i^t-\boldsymbol{w}_t^*\|^2.
\end{equation}

\begin{proof}
Since $L_t$ is $\mu$–strongly convex, for any $\boldsymbol{w},\boldsymbol{w}'\in\mathbb{R}^d$ we have
\begin{align}
L_t(\boldsymbol{w}') \ge L_t(\boldsymbol{w})& + \nabla L_t(\boldsymbol{w})^{T} (\boldsymbol{w}'-\boldsymbol{w}) + \nonumber \\
& \frac{\mu}{2}\|\boldsymbol{w}'-\boldsymbol{w}\|^2, \label{eq:sc1}\\[1mm]
L_t(\boldsymbol{w}) \ge L_t(\boldsymbol{w}')& + \nabla L_t(\boldsymbol{w}')^{T} (\boldsymbol{w}-\boldsymbol{w}') + \nonumber \\
& \frac{\mu}{2}\|\boldsymbol{w}-\boldsymbol{w}'\|^2. \label{eq:sc2}
\end{align}
Adding \eqref{eq:sc1} and \eqref{eq:sc2} yields $0 \ge \nabla L_t(\boldsymbol{w})^{T} (\boldsymbol{w}'-\boldsymbol{w}) + \nabla L_t(\boldsymbol{w}')^{T} (\boldsymbol{w}-\boldsymbol{w}') + \mu \|\boldsymbol{w}'-\boldsymbol{w}\|^2$. Noting that $\nabla L_t(\boldsymbol{w}')^{T} (\boldsymbol{w}-\boldsymbol{w}') = -\nabla L_t(\boldsymbol{w}')^{T} (\boldsymbol{w}'-\boldsymbol{w})$, we deduce $\left(\nabla L_t(\boldsymbol{w}) - \nabla L_t(\boldsymbol{w}')\right)^{T} (\boldsymbol{w}'-\boldsymbol{w}) \ge \mu \|\boldsymbol{w}'-\boldsymbol{w}\|^2$. Setting $\boldsymbol{w} = \boldsymbol{w}_i^t$ and $\boldsymbol{w}' = \boldsymbol{w}_t^*$ completes the proof.
\end{proof}

\subsection{Expert Tracking Theory Proof}\label{appendix_proof_distill_enhance}

\textbf{Lemma \ref{thm:error-bound}:}  (\textbf{Parameterized Norm Inequality.}) Let $a, b$ be elements of a real inner product space and $\alpha > 0$. Then, $\|a+b\|^2 \leq(1+\alpha)\|a\|^2+\left(1+\alpha^{-1}\right)\|b\|^2$.
\begin{proof}
    First expand the left-hand side, $\|a+b\|^2=\|a\|^2+2\langle a, b\rangle+\|b\|^2$. The right-hand side expands to, $(1+\alpha)\|a\|^2+\left(1+\alpha^{-1}\right)\|b\|^2=\|a\|^2+\alpha\|a\|^2+\|b\|^2+\alpha^{-1}\|b\|^2$. Comparing both sides, it suffices to show, $2\langle a, b\rangle \leq \alpha\|a\|^2+\alpha^{-1}\|b\|^2$. This follows from Young's inequality $2xy \leq \alpha x^{2} + \alpha^{-1} y^{2}$ applied to $x = \| a \|$ and $y = \|b\|$, combined with Cauchy-Schwarz inequality $\langle a, b\rangle \leq\|a\|\|b\|$.
\end{proof}

\textbf{Theorem \ref{thm:error-bound}:}  (\textbf{Expert Tracking Theory.}) 
Under Assumptions (A1)--(A3) and with the update rule, $\boldsymbol{w}_i^{t+1} = \boldsymbol{w}_i^t - \eta\, \nabla L_t\bigl(\boldsymbol{w}_i^t\bigr) + \eta\, \gamma\, \epsilon_i^t, \quad \epsilon_i^t \sim \mathcal{N}\left(0, I_d\right)$, there exists a constant $C>0$ (depending only on $L$ and $\mu$) such that the average squared error satisfies, $\frac{1}{T} \sum_{t=1}^{T} \mathbb{E}\|e_t\|^2 \le C \left(\eta \gamma^2 d + \frac{P_T^2}{T \eta}\right)$, where the error is defined as $e_t = \boldsymbol{w}_i^t - \boldsymbol{w}_t^*$ and the path length of the optimal parameters is $P_T = \sum_{t=1}^{T-1}\|\boldsymbol{w}_{t+1}^* - \boldsymbol{w}_t^*\|$.
\begin{proof}
Define the drift of the optimal parameter as $\Delta_t := \boldsymbol{w}_{t+1}^* - \boldsymbol{w}_t^*$. Then the error update is
\begin{align}
& e_{t+1} = \boldsymbol{w}_i^{t+1} - \boldsymbol{w}_{t+1}^* \nonumber\\[1mm]
&= \boldsymbol{w}_i^t - \eta\, \nabla L_t\bigl(\boldsymbol{w}_i^t\bigr) + \eta\, \gamma\, \epsilon_i^t - \boldsymbol{w}_{t+1}^* \nonumber\\[1mm]
&= \underbrace{\left[\boldsymbol{w}_i^t - \boldsymbol{w}_t^* - \eta\Bigl(\nabla L_t\bigl(\boldsymbol{w}_i^t\bigr) - \nabla L_t\bigl(\boldsymbol{w}_t^*\bigr)\Bigr)\right]}_{=: A_t} - \Delta_t + \underbrace{\eta\, \gamma\, \epsilon_i^t}_{=: E_t} \nonumber \\
& =\underbrace{\left(e_t-\eta\left(\nabla L_t\left(\boldsymbol{w}_i^t\right)-\nabla L_t\left(\boldsymbol{w}_t^*\right)\right)\right)}_{\text {Contraction Term } A_t}-\underbrace{\Delta_t}_{\text {Drift Term }}+\underbrace{\eta \gamma \epsilon_i^t}_{\text {Noise Term }}. \nonumber
\end{align}
Thus, $e_{t+1} = A_t - \Delta_t + E_t$.

\textbf{Step 1: Contraction Property of \(A_t\).} \quad  
We now show that under the conditions of \(\mu\)–strong convexity and \(L\)– smoothness (with assumption \(\eta \le 1/L\)), it holds that $\|A_t\|^2 \le (1 - \mu \eta)\,\|e_t\|^2$. 

Indeed, by definition, $A_t = e_t - \eta\Bigl(\nabla L_t(\boldsymbol{w}_i^t)-\nabla L_t(\boldsymbol{w}_t^*)\Bigr)$. Expanding the square norm, we obtain:

\begin{align}
    \|A_t\|^2 = \|e_t\|^2 -  2\eta\, \langle e_t, & \, \nabla L_t(\boldsymbol{w}_i^t)-\nabla L_t(\boldsymbol{w}_t^*)\rangle   \\ 
    & + \eta^2\, \|\nabla L_t(\boldsymbol{w}_i^t)-\nabla L_t(\boldsymbol{w}_t^*)\|^2. \nonumber
\end{align}
Since \(L_t\) is \(\mu\)–strongly convex, it satisfies (Lemma \ref{GradientStrongConvexity})
\[
\langle e_t,\, \nabla L_t(\boldsymbol{w}_i^t)-\nabla L_t(\boldsymbol{w}_t^*)\rangle \ge \mu\, \|e_t\|^2.
\]
Moreover, by \(L\)–smoothness (Lemma \ref{GradientCocoercivityLemma}), $\|\nabla L_t(\boldsymbol{w}_i^t)-\nabla L_t(\boldsymbol{w}_t^*)\|^2 \le L\, \langle e_t,\, \nabla L_t(\boldsymbol{w}_i^t)-\nabla L_t(\boldsymbol{w}_t^*)\rangle$.

Hence, $\eta^2\, \|\nabla L_t(\boldsymbol{w}_i^t)-\nabla L_t(\boldsymbol{w}_t^*)\|^2 \le \eta^2 L\, \langle e_t,\, \nabla L_t(\boldsymbol{w}_i^t)-\nabla L_t(\boldsymbol{w}_t^*)\rangle$. Substituting these bounds back, we have

\begin{align}
\|A_t\|^2 \le \|e_t\|^2 - 2\eta\, \langle e_t,\, & \nabla L_t(\boldsymbol{w}_i^t)-\nabla L_t(\boldsymbol{w}_t^*)\rangle \nonumber \\
& + \eta^2 L\, \langle e_t,\, \nabla L_t(\boldsymbol{w}_i^t)-\nabla L_t(\boldsymbol{w}_t^*)\rangle. \nonumber
\end{align}

That is, $\|A_t\|^2 \le \|e_t\|^2 - \eta\Bigl(2-\eta L\Bigr) \langle e_t,\, \nabla L_t(\boldsymbol{w}_i^t)-\nabla L_t(\boldsymbol{w}_t^*)\rangle$. Using the strong convexity bound \(\langle e_t,\, \nabla L_t(\boldsymbol{w}_i^t)-\nabla L_t(\boldsymbol{w}_t^*)\rangle \ge \mu\, \|e_t\|^2\) leads to $\|A_t\|^2 \le \left[1 - \eta\mu\Bigl(2-\eta L\Bigr)\right]\,\|e_t\|^2$.

Since \(\eta \le \frac{1}{L}\) implies \(2-\eta L \ge 1\), it follows that $\|A_t\|^2 \le (1 - \eta\mu)\,\|e_t\|^2$.

\textbf{Step 2: Recursive Inequality with Noise-Drift Coupling.} \quad  

Expanding $\|e_{t+1} \|^{2}$:
\begin{align}
    \left\|e_{t+1}\right\|^2 & \leq \\ 
    & \left\|A_t-\Delta_t\right\|^2+2 \eta \gamma\left\langle A_t-\Delta_t, \epsilon_i^t\right\rangle+\eta^2 \gamma^2\left\|\epsilon_i^t\right\|^2. \nonumber
\end{align}
Taking expectations (note $\mathbb{E} \left[ \epsilon_{i}^{t} \right] = 0$):
\begin{equation}
    \mathbb{E}\left\|e_{t+1}\right\|^2 \leq \mathbb{E}\left\|A_t-\Delta_t\right\|^2 + \eta^2 \gamma^2 d.
\end{equation}
Applying Lemma \ref{thm:error-bound}: $\|a+b\|^2 \leq(1+\alpha)\|a\|^2+\left(1+\alpha^{-1}\right)\|b\|^2$ (for $\alpha > 0$) to $\mathbb{E}\left\|A_t-\Delta_t\right\|^2$:
\begin{equation}
    \mathbb{E}\left\|e_{t+1}\right\|^2 \leq(1+\alpha)(1-\mu \eta) \mathbb{E}\left\|e_t\right\|^2+\left(1 \! + \! \alpha^{-1}\right)\left\|\Delta_t\right\|^2 \! + \! \eta^2 \gamma^2 d . \nonumber
\end{equation}

Set $\alpha=\mu \eta / 2$ ensuring $(1 + \alpha)(1 - \mu \eta) \leq 1 - \mu \eta / 2$. This yield:
\begin{equation}
    \mathbb{E}\left\|e_{t+1}\right\|^2 \leq\left(1-\frac{\mu \eta}{2}\right) \mathbb{E}\left\|e_t\right\|^2+\frac{2}{\mu \eta}\left\|\Delta_t\right\|^2+\eta^2 \gamma^2 d .
\end{equation}

\textbf{Step 3: Telescoping Sum and Final Bound.} \quad 

Unrolling the recursion over $t$:
\begin{align}
    \mathbb{E}\left\|e_t\right\|^2 \! \leq\left(1-\frac{\mu \eta}{2}\right)^t \mathbb{E}&\left\|e_0\right\|^2 \! +\! \frac{2}{\mu \eta} \! \sum_{k=0}^{t-1}\left(1-\frac{\mu \eta}{2}\right)^{t-1-k}\left\|\Delta_k\right\|^2 \nonumber \\
    &+ \eta^2 \gamma^2 d \sum_{k=0}^{t-1}\left(1-\frac{\mu \eta}{2}\right)^k.
\end{align}
Summing and averaging the Error from $t=1$ to $T$. The initial error term can be expressed as:
\begin{equation}
    \frac{1}{T} \sum_{t=1}^T\left(1-\frac{\mu \eta}{2}\right)^t \mathbb{E}\left\|e_0\right\|^2 \leq \frac{1}{T} \cdot \frac{2}{\mu \eta} \mathbb{E}\left\|e_0\right\|^2.
\end{equation}

For the drift term, interchange summation order and utilizing the geometric series formula:
\begin{align}
    \frac{1}{T} \! \sum_{t=1}^T \! \frac{2}{\mu \eta} \! \sum_{k=0}^{t-1} \! \left(1 \! - \! \frac{\mu \eta}{2} \! \right)^{t \! - \! 1 \! - \!k} \!\left\|\!\Delta_k\!\right\|^2 \! \leq \! \frac{1}{T}\! \cdot\! \frac{4}{\mu^2 \eta^2} \! \sum_{k=0}^{T-1} \! \left\|\!\Delta_k\!\right\|^2.
\end{align}
The noise contribution at each time step follows a geometric series: $\frac{1}{T} \sum_{t=1}^{T} \eta^{2} \gamma^{2} d \sum_{k=0}^{t-1} (1 -\frac{\mu \eta}{2})^{k} \leq \frac{1}{T} \cdot \frac{2 \eta \gamma^{2}dT}{\mu} $. For large \(T\) the initial error term is negligible, and after a refined constant adjustment, we arrive at

\begin{align}
    \frac{1}{T} \sum_{t=1}^{T} \mathbb{E}\|e_t\|^2 & \leq \frac{4P_{T}^{2}}{ \mu^{2} \eta^{2} T} + \frac{2 \eta \gamma^{2} d}{\mu}\\
    & \leq C \left(\eta \gamma^2 d + \frac{P_T^2}{T \eta}\right).
\end{align}

Here, we absorb the coefficients of the higher-order terms into a constant $C$ and implicitly impose the step-size constraint $\frac{4}{\mu^{2}\eta^{2}} \leq \frac{C}{\eta}$. In practical algorithms, the step size is usually small (e.g., $\eta \leq \frac{1}{L}$); thus, this constraint naturally holds in most scenarios. In doing so, we sacrifice some precision but are better able to capture the dynamic balance between the two factors.
\end{proof}



\subsection{Appendix of Impact of Shift QoE on RL}
\label{shift_qoe_qoe_problem}

\begin{figure}[!htbp]
\centering
\includegraphics[width=2.5in]{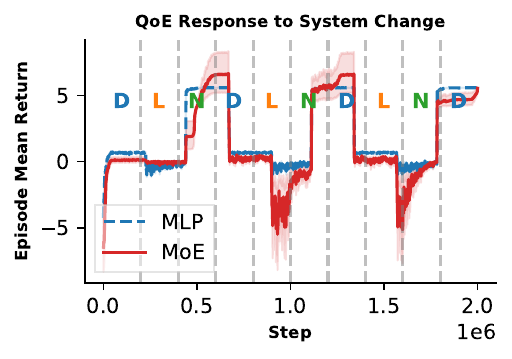}%
\caption{QoE output variation across different network architectures under shifted QoE reward conditions. D, L, and N represent distinct QoE metrics.}
\label{qoe_change_question}
\end{figure}

In this subsection, we explain why we use the mean of the action outputs as the metric for evaluating the MLP's inability to handle QoE shifts. When the objective coefficients change abruptly, overall QoE may fluctuate even if the MLP’s output actions remain unchanged, as shown in Fig. \ref{qoe_change_question}. Such QoE fluctuations are uninformative. For example, if the coefficient of the bitrate term in the QoE function changes from 1 to 6, the QoE contribution of bitrate rises sixfold, even though the policy is effectively unresponsive. Therefore, during objective-function shifts, performance cannot be judged solely by QoE values, as they may convey a false sense of progress. This phenomenon may partly explain why many learning-based methods struggle to generalize under dynamically changing QoE objectives.

\subsection{Impact of QoE Coefficient Variations on Plasticity}
\label{diverse_beta}

In this subsection, we analyze different values of the QoE coefficient $\beta$ to investigate whether the issue of plasticity loss persists under varying QoE preferences. Specifically, we conduct experiments with $\beta = 5$ and $\beta = 7$. As shown in Figure~\ref{beta}, the MLP network produces nearly identical action outputs across different $\beta$ values, indicating a lack of plasticity. This demonstrates that the problem of plasticity loss exists regardless of the specific QoE coefficient used.

\begin{figure}[!htbp]
\centering
\includegraphics[width=2.5in]{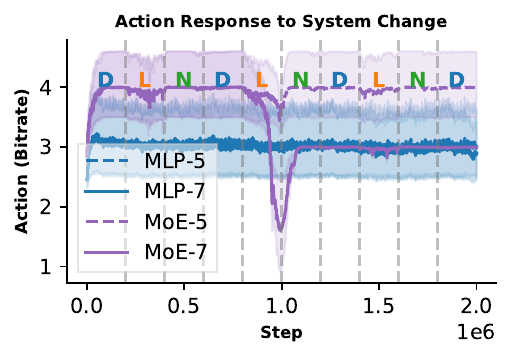}%
\caption{Action output variation across different $\beta$.}
\label{beta}
\end{figure}

\subsection{Algorithm pseudocode}
\label{Algorithm pseudocode}

To facilitate implementation, Algorithm \ref{alg:ppo_pamoe} provides complete pseudocode with explicit definitions of routing and plasticity-aware noise update rules.

\begin{algorithm}
\caption{PPO with Plasticity Aware MoE}
\label{alg:ppo_pamoe}

\KwIn{Env $env$; policy params $\pi_{\omega}$ (router + actor experts); value params $V_{\omega}$ (router + critic experts); rollout length $T$; update epochs $K$.}
\KwOut{Optimized policy $\boldsymbol{\omega}$}

\For{iteration $=1,2,\dots$}{
  Reset trajectory buffer; $\mathcal{B} \gets \varnothing$
  
  \For{$t=1,\dots,T$}{
    Observe state $s_t$\;

    Compute router features $h^{\pi_{\omega}}_t \gets f_{\text{router}}^{\pi_{\omega}}(s_t)$\;
    Compute clean logits $u^{\pi_{\omega}} \gets W^{\pi_{\omega}}_{\text{topk}} h^{\pi_{\omega}}_t$\;
    Add router noise scale $z^{\pi_{\omega}}$\;
    $n^\star \gets \arg\max_j \big(u^{\pi_{\omega}}_j + z^{\pi_{\omega}}_j\big)$\tcp*{Top-1 expert index}
    $e^{\pi_{\omega}}_t \gets \text{Expert}^{\pi_{\omega}}_{n^\star}(s_t)$ \tcp*{selected actor expert}
    $\ell_t \gets W_{\text{policy}} e^{\pi_{\omega}}_t$\tcp*{action logits}
    Sample $a_t \sim \text{Categorical}(\text{logits}=\ell_t)$ and save $\log\pi^{\omega}(a_t|s_t)$\;

    Compute value $h^{V_{\omega}}_t,\,u^{V_{\omega}},\,z^{V_{\omega}}$ analogously\;
    $m^\star \gets \arg\max_j \big(u^{V_{\omega}}_j + z^{V_{\omega}}_j\big)$\;
    $e^{V_{\omega}}_t \gets \text{Expert}^{V_{\omega}}_{m^\star}(s_t)$\;
    $v_t \gets W_{\text{value}} e^{V_{\omega}}_t$\;

    Step env with $a_t$ to get $r_t, s_{t+1}, d_{t+1}$\;
    Push $(s_t,a_t,r_t,d_{t+1},\log\pi_\theta(a_t|s_t), V_t)$ into $\mathcal{B}$\;
    \If{$d_{t+1}$}{reset env and continue}
  }
  
  Compute advantages $A_t$ and returns $R_t$ from $\mathcal{B}$;

  \For{epoch $=1,\dots,K$}{
    \For{each minibatch $B\subset \mathcal{B}$}{
      Update $\boldsymbol{w}$ using noise injection by Eq. (\ref{ppoloss})\;
    }
  }
  
  (optional) Save checkpoints periodically\;
}
\end{algorithm}

\subsection{Expert Contributions to Mitigating Plasticity Loss in PA-MoE} \label{Expert_Contributions_to_Mitigating_Plasticity_Loss_in_PA_MoE}

In this subsection, we introduce the Dormant Neuron Ratio as a quantitative measure of plasticity loss in neural networks. It captures the proportion of inactive neurons during training, reflecting reduced expressiveness and adaptability, especially in continual or multitask learning. Tracking neuron dormancy is thus essential for assessing model capacity utilization in complex modular architectures such as PA-MoE. To formalize this, we introduce the following definition:


\begin{definition}[\textbf{Dormant Neuron}] \label{dormant_neuron} Consider a neural network layer $l$ consisting of $H_l$ neurons. We define the dormancy index $s_{l,i}$ of the $i$-th neuron in layer $l$ as $s_{l,i} = \frac{\mathbb{E}_{\mathbf{x}\in D}|h_{l,i}(\mathbf{x})|}{\frac{1}{H_{l}}\sum_{j \in h}\mathbb{E}_{\mathbf{x}\in D}|h_{l, j}(\mathbf{x})|}$, where $h_{l,i}(\mathbf{x})$ denotes the activation of neuron $(l,i)$ given input $\mathbf{x}$, and $\mathbb{E}_{\mathbf{x} \in D}$ denotes the expectation over input samples drawn from the data distribution $D$.
\end{definition}

A neuron is deemed dormant if its dormancy index $s_{l,i}$ is below a threshold, i.e., its average activation is much lower than the layer’s overall average. A high dormant ratio indicates plasticity loss, with many neurons contributing little to the model’s output.


The experimental results are visualized in Figure~\ref{dormantandzerogradient}, where we plot the Layer Dormant (LD) ratios for each expert in the PA-MoE framework. These visualizations cover both the policy and value networks and provide an insightful cross-layer analysis.

\begin{figure}[!htbp]
\centering
\includegraphics[width=\linewidth]{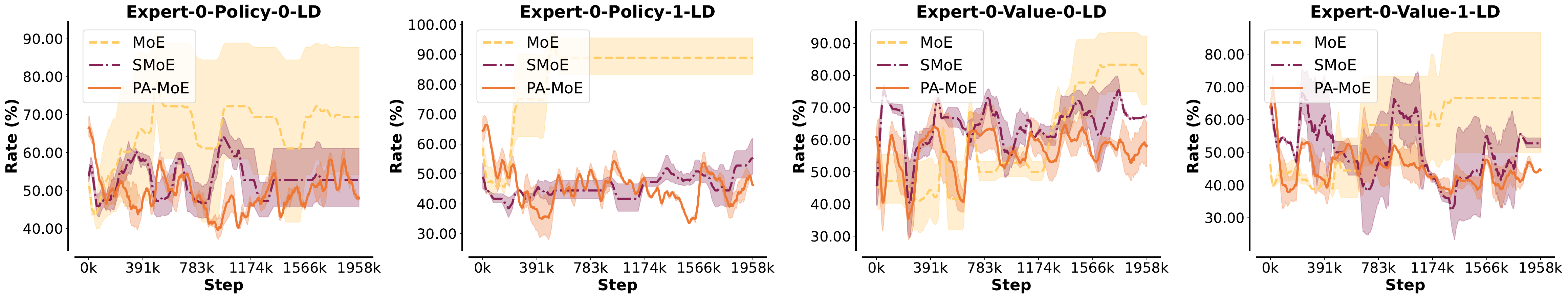}
\includegraphics[width=\linewidth]{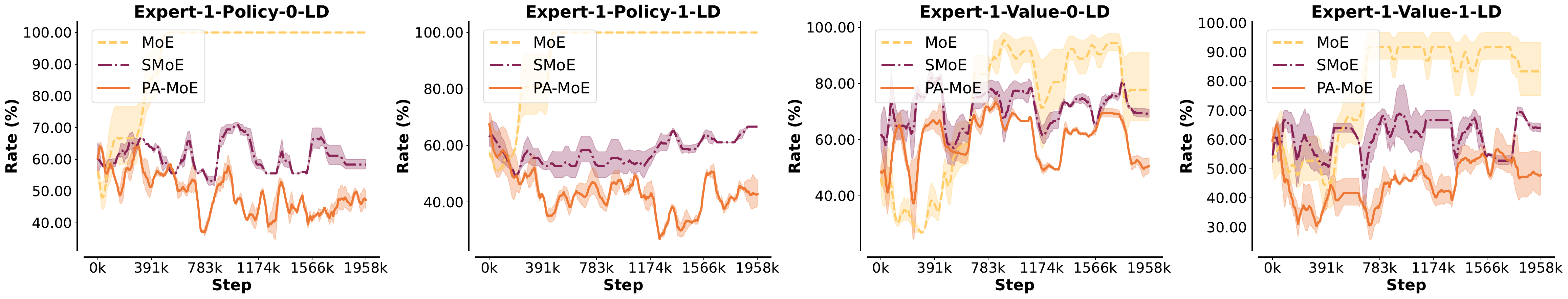}
\includegraphics[width=\linewidth]{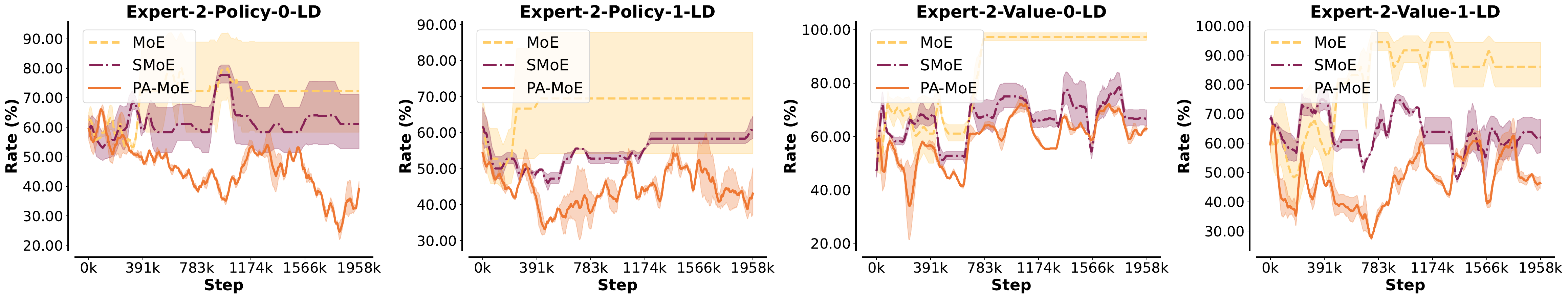}
\caption{Visualization results of the dormant neuron (Layer Dormant, LD) ratios across different layers of the policy and value networks for various experts.}
\label{dormantandzerogradient}
\end{figure}

As shown in the Figure~\ref{dormantandzerogradient}, PA-MoE consistently maintains a lower proportion of dormant neurons across different experts and layers compared to baseline models. This suggests that PA-MoE’s architectural design—particularly its noise-injection mechanism and adaptive routing—helps in better utilization of network capacity, encouraging more balanced activation across neurons and reducing the risk of premature convergence or feature stagnation.

\subsection{Layer-wise Analysis of Expert Rank in PA-MoE} \label{rank}

To better understand plasticity loss, we analyze the feature rank of hidden activations in addition to dormant neurons. While the Dormant Neuron Ratio shows how many neurons are inactive, it does not indicate whether active neurons produce diverse or redundant outputs. Therefore, we propose the Rank of Batch Hidden Layer Output as a complementary diagnostic of representation richness.



\begin{definition}[{\textbf{Rank of Batch Hidden Layer Output}}]
\label{rank_define} Let a hidden layer consist of $H$ neurons. For a batch of inputs of size $B$, the layer's output forms a matrix $M \in \mathbb{R}^{B \times H}$, where each row represents the activation vector of a single sample. We define the rank of $M$ as the number of linearly independent rows (or columns), which is equivalent to the number of nonzero singular values obtained via singular value decomposition (SVD). Formally, if $\{\sigma_{i}\}$ are the singular values of $M$, then $\operatorname{rank}(M)=\#\left\{i \mid \sigma_i>0\right\}$. This metric measures the diversity and dimensionality of hidden representations, indicating the model’s ability to capture distinct input features.

\end{definition}

\begin{figure}[!htbp]
\centering
\includegraphics[width=\linewidth]{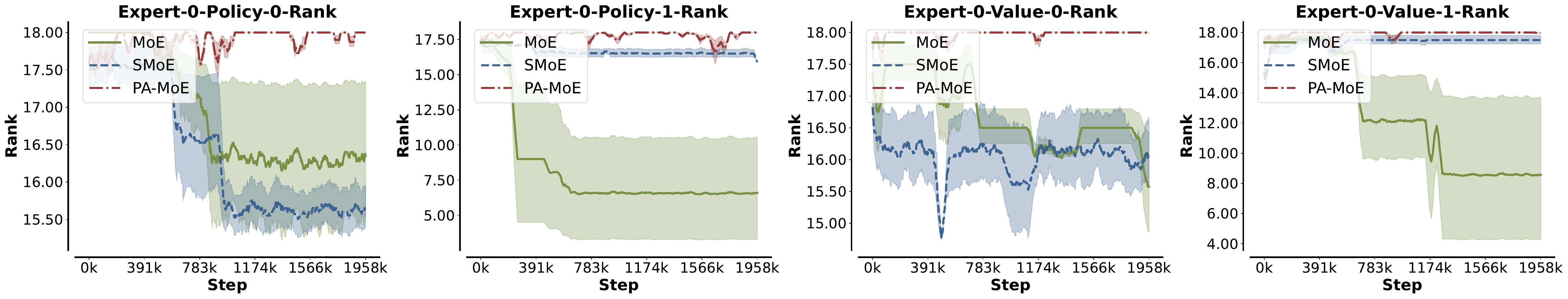}
\includegraphics[width=\linewidth]{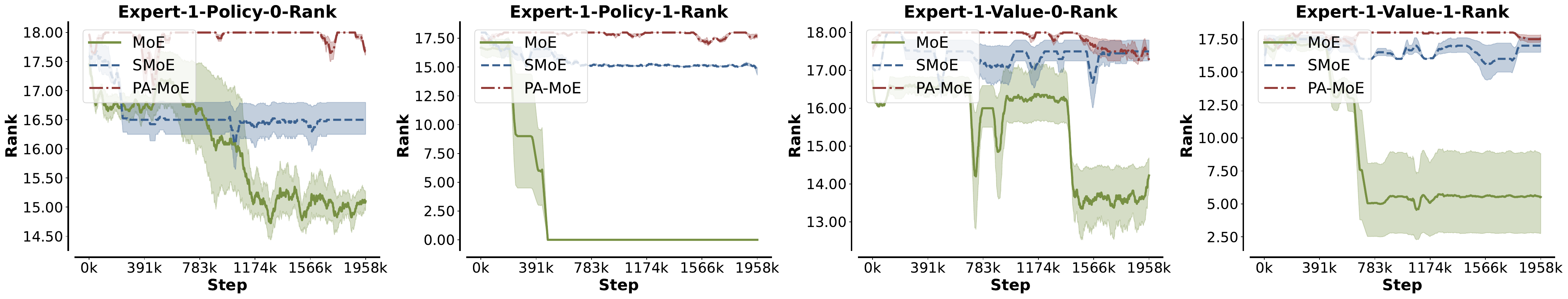}
\includegraphics[width=\linewidth]{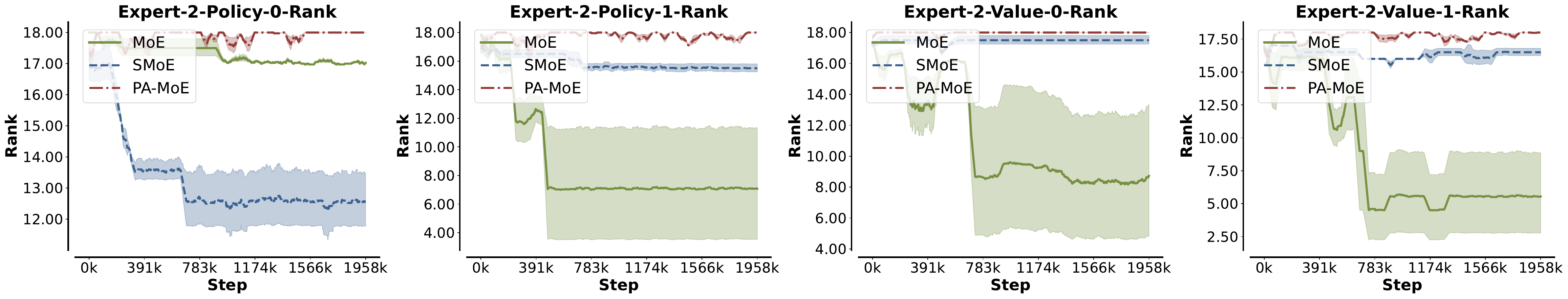}
\caption{Visualization results of the feature rank (Rank) across different layers of the policy and value networks for various experts.}
\label{rank_fig}
\end{figure}

Figure~\ref{rank_fig} illustrates the evolution of feature rank across different experts under various algorithms. It is evident that our approach, PA-MoE, consistently maintains a higher rank, effectively mitigating the loss of plasticity.

\subsection{Layer-wise Analysis of Expert Effective Rank in PA-MoE} \label{effective_rank}

To further understand the representational capacity of different models, we analyze the effective rank of hidden features across the network layers. The effective rank serves as a proxy for the diversity and richness of learned representations. By evaluating this metric, we gain insights into how well each method utilizes the available feature space.

\begin{definition}[{\textbf{Effective Rank}}]
\label{eff_rank} Let $\{\sigma_{i}\}_{i=1}^{n}$ be the singular values of a matrix. Define the normalized singular values as $p_{i} = \frac{|\sigma_{i}|}{\sum_{j=1}^{n}|\sigma_{j}|}$. The Shannon entropy of these normalized values is computed as $H = - \sum_{i=1}^{n} p_{i} \log(p_{i})$, with the convention that $0 \log0=0$. The effective rank is then defined as $r_{\text{eff}} = e^{H}$. This measure reflects the effective dimensionality of the matrix, capturing the spread and diversity of its singular value distribution.
\end{definition}

\begin{figure}[!htbp]
\centering
\includegraphics[width=\linewidth]{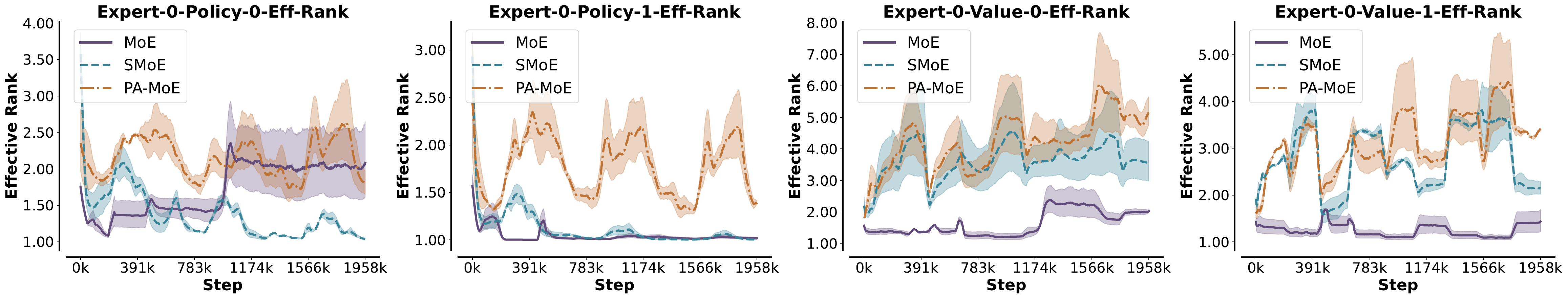}
\includegraphics[width=\linewidth]{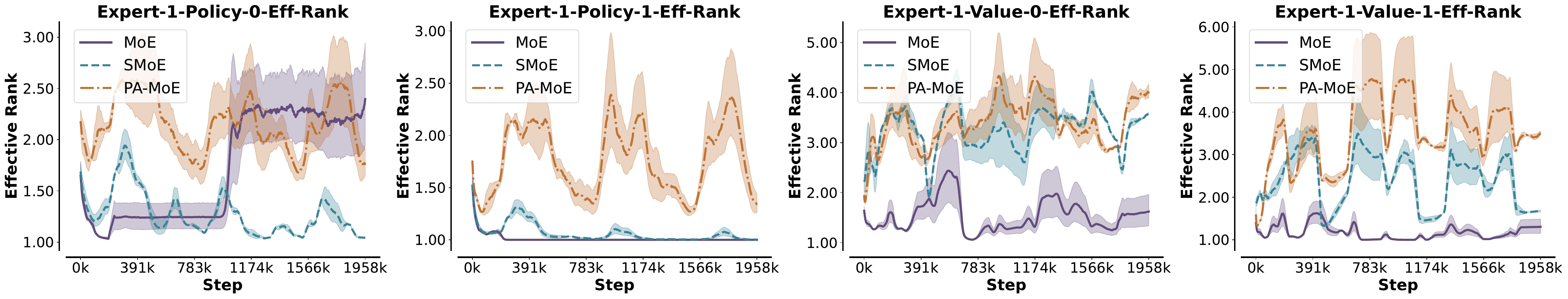}
\includegraphics[width=\linewidth]{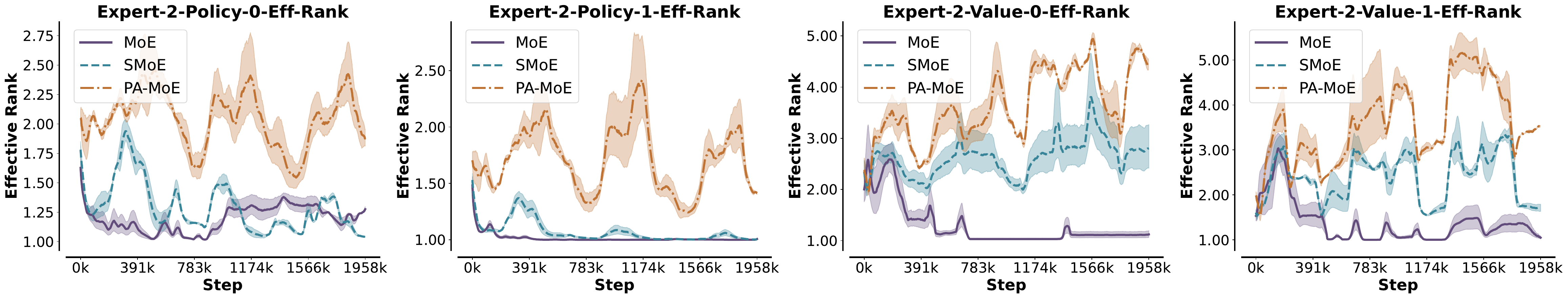}
\caption{Visualization results of the effective rank (Eff-Rank) across different layers of the policy and value networks for various experts.}
\label{eff_rank_fig}
\end{figure}

Figure~\ref{eff_rank_fig} shows that our method consistently yields a higher effective rank than others. This rank reflects the intrinsic dimensionality of hidden representations, with higher values indicating richer and more diverse features. By preserving plasticity, our approach enables more effective adaptation to varying conditions and captures a broader range of data characteristics.

\subsection{Layer-wise Analysis of Expert Approximate Rank in PA-MoE} \label{approx_rank}

To complement effective rank, we analyze the approximate rank across layers. While effective rank captures entropy-based diversity, approximate rank highlights variance concentration, showing how many principal directions significantly shape the feature space and how compact yet expressive the representations are under different expert configurations.


\begin{definition}[{\textbf{Approximate Rank}}]
Let $\{\sigma_{i}\}_{i=1}^{n}$ be the singular values of the feature matrix $M$. The squared singular values $\sigma_{i}^{2}$ represent the variance contributed by each singular direction. Normalize these values as $p_{i}=\frac{\sigma_{i}^{2}}{\sum_{j=1}^{n}\sigma^{2}}$.

The approximate rank is the smallest $r$ such that $\sum_{i=1}^{r}p_{i} \geq \text{prop}$ (typically 0.99), capturing the minimal number of singular components needed to explain a specified proportion of variance and thus quantifying the effective dimensionality of $M$.

\end{definition}

\begin{figure}[!htbp]
\centering
\includegraphics[width=\linewidth]{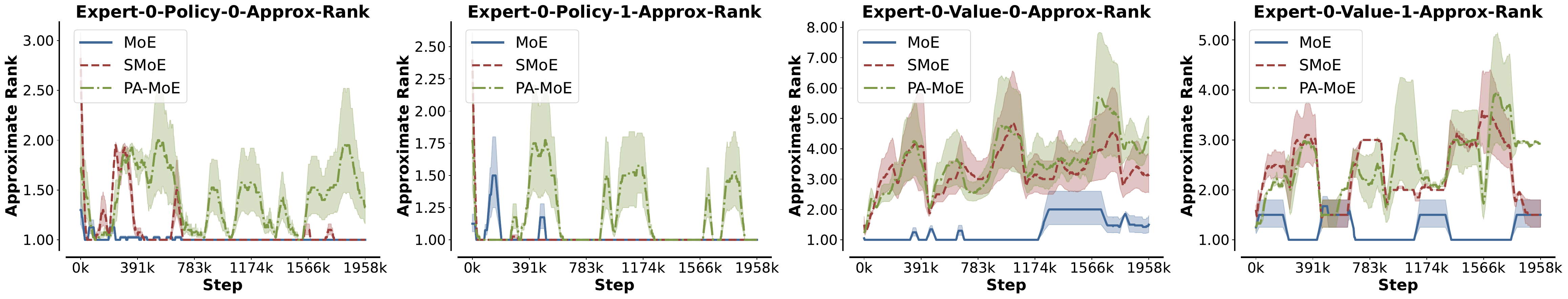}
\includegraphics[width=\linewidth]{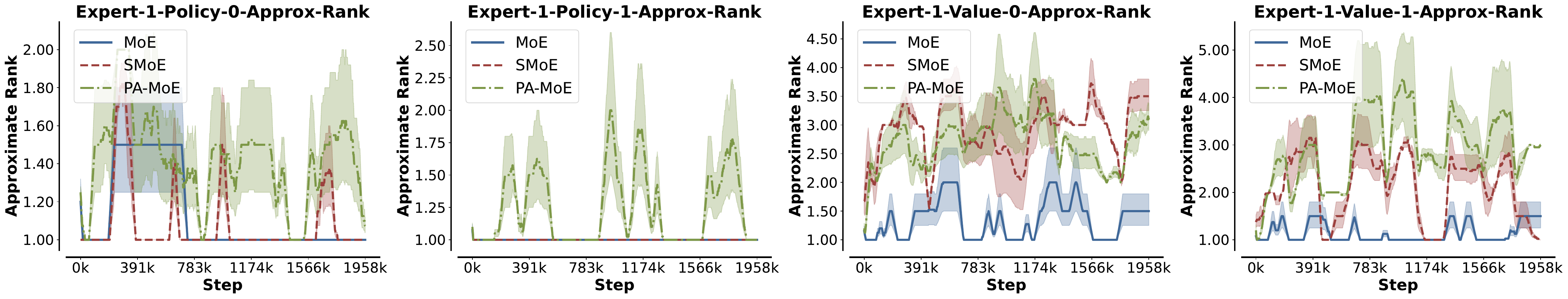}
\includegraphics[width=\linewidth]{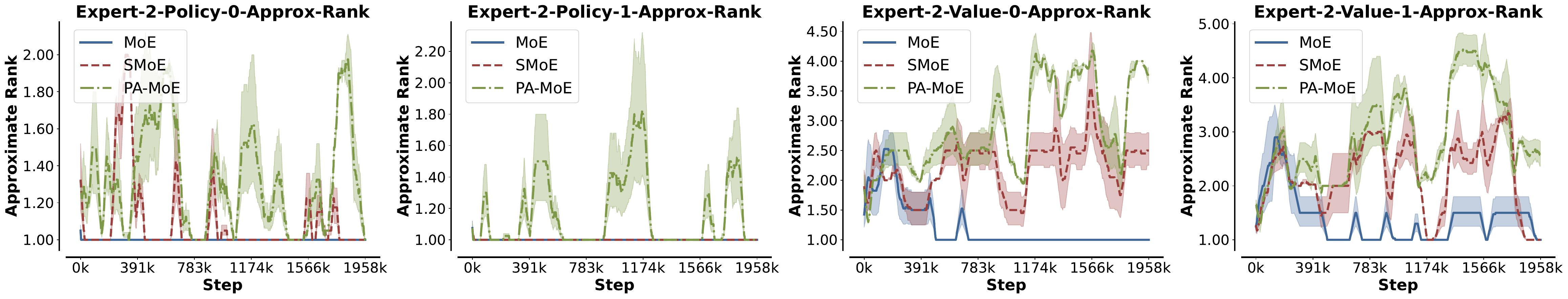}
\caption{Visualization results of the approximate rank (Approx-Rank) across different layers of the policy and value networks for various experts.}
\label{approx_rank_fig}
\end{figure}

Figure~\ref{approx_rank_fig} shows that our proposed PA-MoE approach consistently achieves a higher approximate rank than other methods, reflecting its ability to capture a greater proportion of the feature variance and thereby better preserve model plasticity.

\subsection{Visualization of Expert Weight Distribution in PA-MoE}
\label{Expert_weight_to_Mitigating_Plasticity_Loss_in_PA_MoE}

\begin{figure}[!htbp]
\centering
\includegraphics[width=\linewidth]{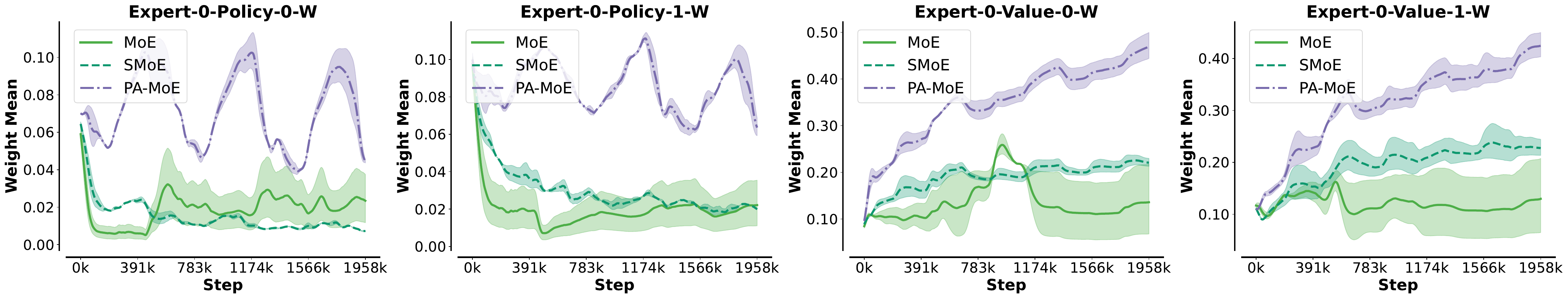}
\includegraphics[width=\linewidth]{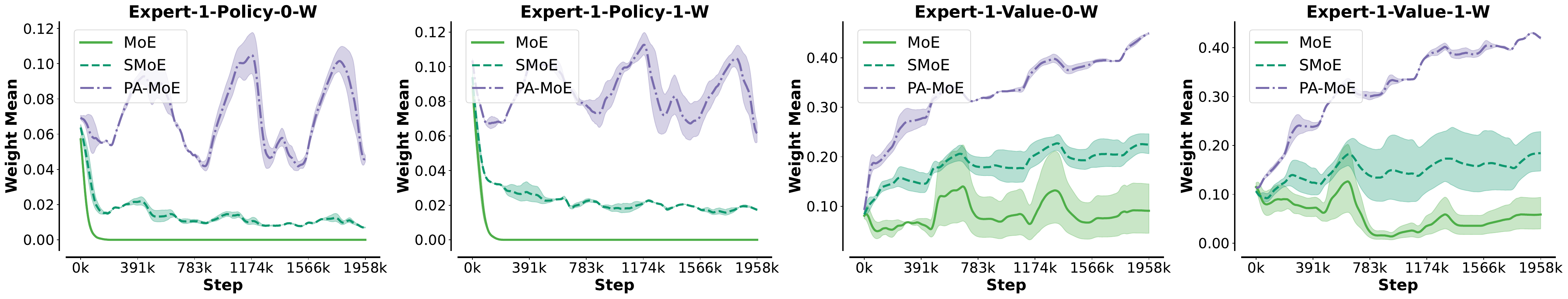}
\includegraphics[width=\linewidth]{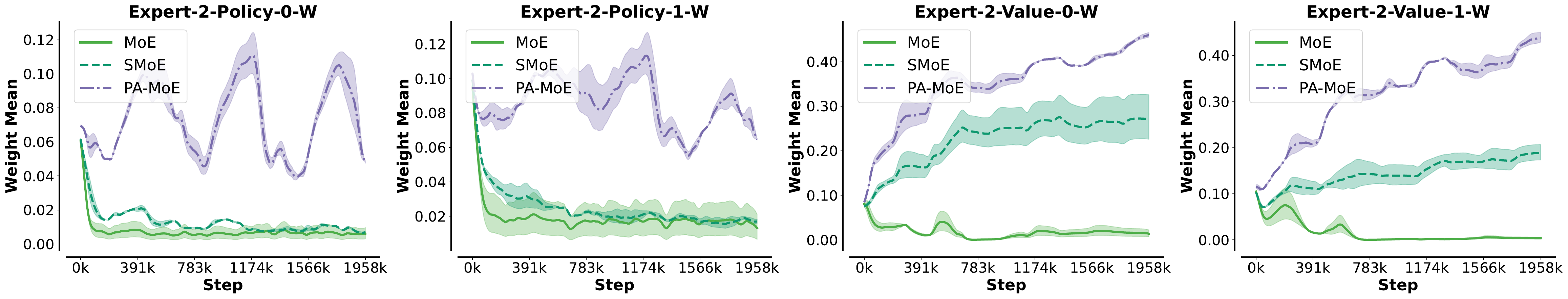}
\caption{Visualization results of the weight mean (W) across different layers of the policy and value networks for various experts.}
\label{weight}
\end{figure}

Figure~\ref{weight} visualizes weight dynamics across layers and experts under QoE shifts. Compared to the baseline, PA-MoE shows consistently larger weights due to its noise-injection-based forgetting mechanism: the network amplifies critical weights to offset noise, maintaining a balance between retention and forgetting essential for plasticity in non-stationary environments.


\end{document}